%% file: NoGoResult.tex
\documentclass[a4paper,10pt]{article}
\usepackage{graphicx}
\usepackage{amsmath}
\usepackage{citesort}
\usepackage{amsfonts}
\usepackage{amssymb}
\usepackage{amsmath}
\usepackage{latexsym}
\usepackage{color}
\usepackage{ntheorem}
\usepackage{braket}
\usepackage{url}
\input{LeeNewcommand2}

\DeclareFontFamily{OT1}{rsfs}{}
\DeclareFontShape{OT1}{rsfs}{m}{n}{ <-7> rsfs5 <7-10> rsfs7 <10-> rsfs10}{}
\DeclareMathAlphabet{\mycal}{OT1}{rsfs}{m}{n}

{\catcode `\@=11 \global\let\AddToReset=\@addtoreset}
\AddToReset{equation}{section}

{\catcode `\@=11 \global\let\AddToReset=\@addtoreset}
\AddToReset{figure}{section}

{\catcode `\@=11 \global\let\AddToReset=\@addtoreset}
\AddToReset{table}{section}

\begin{document}

\title{Characteristic initial data and smoothness of Scri.\\ I. Framework and results%
\thanks{Preprint UWThPh-2014-1.}}
\author{
Piotr T. Chru\'sciel{}\thanks{Email  {Piotr.Chrusciel@univie.ac.at}, URL {
http://homepage.univie.ac.at/piotr.chrusciel}}\ \ and
Tim-Torben Paetz{}\thanks{Email  Tim-Torben.Paetz@univie.ac.at}   %\vspace{0.5em}  \\  University of Vienna
 \vspace{0.5em}\\  \textit{Gravitational Physics, University of Vienna}  \\ \textit{Boltzmanngasse 5, 1090 Vienna, Austria }}
\maketitle

\vspace{-0.2em}

\begin{abstract}
We analyze the Cauchy problem for the vacuum
 %x% \ptcr{vacuum added}
Einstein equations with data on a complete light-cone in an asymptotically Minkowskian space-time. We provide conditions on the free initial data which guarantee existence of global solutions of the characteristic constraint equations. We present necessary-and-sufficient conditions on characteristic
 initial data
in $3+1$ dimensions to have no logarithmic terms in an asymptotic expansion at  null infinity.
\end{abstract}

\noindent
\hspace{2.1em} PACS: 02.30.Hq, 02.30.Mv, 04.20.Ex, 04.20.Ha

\tableofcontents

\section{Introduction}

%\ptcr{rewrites; old introduction moved to a file oldintro.tex}
An issue of central importance in general relativity is the understanding of gravitational radiation. This has direct implications for the soon-expected direct detection of gravitational waves. The current main effort in this topic appears to be a mixture of numerical modeling and approximation methods. From this perspective there does not seem to be a need for a better understanding of the exact properties of the gravitational field in the radiation regime. However, as observations and numerics will have become routine, solid theoretical foundations for the problem will become necessary.

Now, a generally accepted framework for describing gravitational radiation seems to be the method of conformal completions of Penrose. Here a key hypothesis is that a suitable conformal rescaling of the space-time metric becomes smooth on a new manifold with boundary $\scrip$.
  One then needs to face the question, if and how such space-times can be constructed. Ultimately one would like to isolate the class of initial data, on a spacelike slice extending to spatial infinity,  the evolution of which admits a Penrose-type conformal completion at infinity, and show that the class is large enough to model all physical processes at hand. Direct attempts to carry this out  (see~\cite{Friedrich:tuebingen,Kroon6,Valiente-Kroon:2002fa} and references therein) have not been successful so far. Similarly, the asymptotic behaviour of the gravitational field established in~\cite{Ch-Kl,KlainermanNicoloBook,klainerman:nicolo:review,KlainermanNicoloPeeling,BieriZipser,LindbladRodnianski} is inconclusive as far as the smoothness of the conformally rescaled metric at $\scrip$ is concerned.  The reader is referred to~\cite{FriedrichCMP13} for an extensive discussion of the issues arising.
%x% \ptcr{comment and ref added as suggested by the referee}

   On the other hand, clear-cut constructions have been carried-out in less demanding settings,  with data on characteristic surfaces as pioneered by Bondi et al.~\cite{BBM}, or with initial data with hyperboloidal asymptotics. It has been found~\cite{TorrenceCouch,ChMS,andersson:chrusciel:PRL,ACF}
%x%    \ptcr{ref added as suggested by the referee}
    that both generic Bondi data and generic hyperboloidal data, constructed out of conformally smooth seed data, will \emph{not} lead to space-times with a smooth conformal completion. Instead, a \emph{polyhomogeneous} asymptotics of solutions of the relevant constraint equations was obtained, with logarithmic terms appearing in asymptotic expansions of the fields.

The case for the necessity of a polyhomogeneous-at-best framework, as resulting from the above work, is not waterproof:  In both cases  it is not clear whether initial data with logarithmic terms can arise from evolution of a physical system which is asymptotically flat in spacelike directions. There is a further issue with the Bondi expansions, because the framework of Bondi et al.~\cite{BBM,Sachs} does not provide a well-posed system of evolution equations for the problem at hand.

The aim of this work is to rederive the existence of obstructions to smoothness of the metric at $\scrip$ in a framework in which the evolution problem for the Einstein vacuum equations is well-posed and where free initial data are given on a light-cone extending to null infinity, or on two characteristic hypersurfaces one of which extends to infinity, or in a mixed setting where part of the data are prescribed on a spacelike surface and part on a characteristic one extending to infinity.
This can be viewed as a revisiting of the Bondi-type setting in a framework where an associated space-time is guaranteed to exist.

One of the attractive features of the characteristic Cauchy problem is that one can explicitly provide an exhaustive class of freely prescribable initial data. By ``exhaustive class" we mean that the map from the space of free initial data to the set of solutions is surjective, where ``solution'' refers to that part of space-time which is covered by the domain of dependence of the smooth part of the light-cone, or of the smooth part of the null hypersurfaces issuing normally from a smooth submanifold of codimension two.%
\footnote{This should be contrasted with the spacelike Cauchy problem, where no exhaustive method for constructing non-CMC initial data sets is known. It should, however, be kept in mind that the spacelike Cauchy problem does not suffer from the serious problem of formation of caustics, inherent to the characteristic one.}
%
%x% \ptcr{comment added, as hinted-to by the referee; footnote reworded}
There is, moreover, considerable flexibility in prescribing characteristic initial data~\cite{ChPaetz}. In this work we will concentrate on the following  approaches:

\begin{enumerate}
  \item The free data are a triple $(\mcN,[\gamma], \kappa)$,
    where $\mcN$ is a $n$-dimensional manifold,   $[\gamma]$  is a conformal class of symmetric two-covariant tensors on $\mcN$ of signature $(0,+,\ldots,+)$, and $\kappa$ is  a field of connections on the bundles of tangents to the integral curves of the kernel of $\gamma$.%
    \footnote{Recall that a connection $\nabla$ on each such bundle is uniquely described by
   writing
  $
   \nabla_r \partial_r = \kappa \partial_r
  $, in a coordinate system where $\partial_r$ is in the kernel of $\gamma$.
  Once the associated space-time has been constructed we will also have
  $
   \nabla_r \partial_r = \kappa \partial_r
  $,
  where $\nabla$ now is the covariant derivative operator associated with the space-time metric.}
%x%   \ptcr{footnote to make things more precise, as requested by the referee}

  \item Alternatively, the  data are a triple $(\mcN,\check g, \kappa)$, where  $\check g$ is a field of symmetric two-covariant tensors on $\mcN$ of signature $(0,+,\ldots,+)$, and $\kappa$ is  a field of connections on the bundles of tangents to the integral curves of the characteristic direction of $\check g$.
       \footnote{We will often write $(\check g, \kappa)$ instead of $(\mcN,\check g, \kappa)$, with $\mcN$ being implicitly understood, when no precise description of $\mcN$ is required.}
       %
  %x%     \ptcr{footnote added for clarification as requested by the referee/}
      The pair  $(\check g, \kappa)$ is further required to satisfy the constraint equation
  \bel{10XII13.2}
    \partial_r\tau - \kappa \tau + |\sigma|^2 + \frac{\tau^2}{n-1} = 0
  \;,
  \ee
  where $\tau$ is the divergence and $\sigma$ is the shear (see Section~\ref{kappa_freedom} for details), which will be referred to as the \emph{Raychaudhuri equation}.

  \item Alternatively, the connection coefficient $\kappa$ and all the components of the space-time metric are prescribed on $\mcN$, subject to the Raychaudhuri constraint equation.
  %x%     \ptcr{clarification added, as requested by the referee}
       Here $\mcN$ is viewed as the hypersurface $\{u=0\}$ in the space-time to-be-constructed,
       and thus all metric components $g_{\mu\nu}$ are prescribed at $u=0$ in a coordinate system $(x^\mu)=(u,x^i)$, where $(x^i)$ are local coordinates on $\mcN$.

     \item Finally, schemes where tetrad components of the conformal Weyl tensor are used as
      free data are briefly discussed.
%x% \ptcr{reworded, as requested by the referee}
\end{enumerate}

In the first two cases, to obtain a well posed evolution problem one needs to impose gauge conditions; in the third case,
the initial data themselves determine the gauge, with the ``gauge-source functions''  determined from the initial data.

%x% \ptcr{clarification as requested by the referee}
The aim of this work is to analyze the occurrence of log terms in the asymptotic expansions as $r$ goes to infinity for initial data sets as above.
The gauge choice $\kappa=O(r^{-3})$ below (in particular the gauge choice $\kappa=\frac{r}{2}|\sigma|^2$,
on which we focus in part II \cite{TimAsymptotics},
ensures that  affine parameters along the generators of $\mcN$  diverge as $r$ goes to infinity (cf.\ \cite[Appendix~B]{TimAsymptotics}), so that in the associated space-time the limit  $r\to\infty$  will correspond to null geodesics approaching  a (possibly non-smooth) null infinity.

It turns out that the simplest choice of gauge conditions, namely $\kappa=0$ and harmonic coordinates, is \emph{not compatible} with  smooth asymptotics at the conformal boundary at infinity: we prove that the \emph{only} vacuum metric, constructed from characteristic Cauchy data on a light-cone, and which has a smooth conformal completion in this gauge, is Minkowski space-time.

%x%\ptcr{remark and reference added as suggested by the referee}
{It should be pointed out, that the observation that some sets of harmonic coordinates are problematic for an analysis of null infinity has already been made in \cite{ChoquetBruhat73,BlanchetPRSL87}. Our contribution here is to make a precise \emph{{no-go}} statement, without approximation procedures or supplementary assumptions.}

%x% \ptcr{rewordings}
One way out of the problem is to replace the harmonic-coordinates condition by a wave-map gauge with non-vanishing gauge-source functions. This provides a useful tool to isolate those log terms which are gauge artifacts, in the sense that they can be removed from the solution by an appropriate choice of the gauge-source functions. There remain, however, some logarithmic coefficients which cannot be removed in this way. We identify those coefficients, and show that the requirement that these coefficients do not vanish is  gauge-independent. In part~II of this work
we show that the logarithmic coefficients are non-zero for generic initial data.
 The equations which lead to vanishing logarithmic coefficients will be referred to as the
\emph{no-logs-condition}.

It is expected that for generic initial data sets, as considered here, the space-times obtained by solving the Cauchy problem will have a polyhomogeneous expansion at null infinity.
%x% \ptcr{wording improved, as suggested by the referee}
There are, however, no  theorems in the existing mathematical literature which guarantee existence of a polyhomogeneous $\scrip$ when the  initial data have non-trivial log terms.
% \ptc{should be done...}

The situation is different when the no-logs-condition is satisfied. In part~II of this work
we show that the resulting initial data lead to smooth initial data for Friedrich's conformal field equations~\cite{F1} as considered in~\cite{CPW}. This implies that the no-logs-condition provides a necessary-and-sufficient condition for the evolved space-time to posses a smooth $\scrip$. For initial data close enough to Minkowskian ones, solutions global to the future are obtained.

It may still be the case that the logarithmic expansions are irrelevant as far as our understanding of gravitational radiation is concerned, either because they never arise from the evolution of isolated physical systems, or because their occurrence prevents existence of a sufficiently long evolution of the data, or because all essential physical issues are already satisfactorily described by smooth conformal completions. While we haven't provided a definite answer to those questions, we hope that our results here will contribute to resolve the issue.

%x%\ptcr{added, as requested by the referee}
If not explicitly stated otherwise, all manifolds, fields, and expansion coefficients are assumed to be smooth.

\section{The characteristic Cauchy problem on a light-cone}
\label{cauchy_problem}

In this section we will review some facts concerning the characteristic Cauchy problem. Most of the discussion applies to any characteristic surface. We concentrate on a light-cone,
as in this case all the information needed is contained in the characteristic initial data together with the requirement of the smoothness
of the metric at the vertex. The remaining Cauchy problems mentioned in the Introduction will be discussed in Section~\ref{s16XII13.1}  below.

\subsection{Gauge freedom}
\subsubsection{Adapted null coordinates}
\label{Adapted null coordinates}

Our starting point is a $C^{\infty}$-manifold $\mcM \cong\mathbb{R}^{n+1}$ and a future light-cone $C_O\subset \mcM$ emanating from some point $O\in \mcM$.
We make the assumption that the subset $C_O$ can be \textit{globally} represented in suitable coordinates $(y^{\mu})$  by the equation of a Minkowskian cone, i.e.\
\begin{equation*}
 C_O = \{ (y^{\mu}) :  y^0 = \sqrt{\sum_{i=1}^n (y^i)^2} \} \subset \mcM
 \;.
\end{equation*}
Given a $C^{1,1}$-Lorentzian space-time
such a representation is always possible in some neighbourhood of the vertex.
However, since caustics may develop along the null geodesics which generate the cone,
it is a geometric restriction to assume the existence of a Minkowskian representation globally.

A treatment of the characteristic initial value problem at hand is easier in coordinates $x^{\mu}$ adapted to the geometry of the light-cone~\cite{RendallCIVP,CCM2}. We consider space-time-dimensions $n+1\geq 3$.
It is standard to construct a set of coordinates  $(x^{\mu})\equiv(u,r,x^A)$, $A=2,\dots,n$, so that $C_O\setminus\{0\}=\{u = 0\}$.
%x%\ptcr{clarifications, as requested by the referee}
The $x^A$'s denote local coordinates on the level sets $\Sigma_r:=\{r=\text{const},u=0\}\cong S^{n-1}$, and are constant along the generators. The coordinate $r$ induces, by restriction, a parameterization of the generators and is chosen so that the point $O$ is approached when $r \rightarrow 0$.
The general form of the trace $\overline g$ on the cone $C_O$ of the space-time metric $g$ reduces in these \textit{adapted null coordinates} to
\begin{equation}
 \overline g = \overline g_{00}\mathrm{d}u^2 + 2\nu_0\mathrm{d}u\mathrm{d}r + 2\nu_A\mathrm{d}u\mathrm{d}x^A + \coneg
 \;,
 \label{null2}
\end{equation}
where
\begin{equation*}
 \nu_0:=\overline g_{01}\;, \quad \nu_A:=\overline g_{0A}\;,
\end{equation*}
and where
% \ptcr{problematic notation $\tilde g$, make it $\check g$?}
%
\begin{equation*}
 \coneg = \coneg_{AB}\mathrm{d}x^A\mathrm{d}x^B := \overline g_{AB}\mathrm{d}x^A\mathrm{d}x^B
\end{equation*}
is a degenerate quadratic form induced by $g$ on $C_O$ which induces on each slice $\Sigma_r$ an $r$-dependent Riemannian metric  $\coneg_{\Sigma_r}$ (coinciding with $\coneg (r,\cdot)$ in the coordinates above).%
\footnote{The degenerate quadratic form denoted here by $\coneg$ has been denoted by $\tilde g$ in~\cite{CCM2,ChPaetz}. However, here we will use~$\tilde g$ to denote the conformally rescaled unphysical metric, as done in most of the literature on the subject.}

The components $\overline g_{00}$, $\nu_0$ and $\nu_A$ are gauge-dependent quantities.
 %x%\ptcr{added, as requested by the referee}
In particular, $\nu_0$ changes sign when $u$ is replaced by $-u$. Whenever useful and/or relevant, we will assume that $\partial_r$ is future-directed and $\partial_u$ is past-directed, which corresponds to requiring that $\nu_0>0$.

The quadratic form $\coneg$ is intrinsically defined on $C_O$,
independently of the choice of the parameter $r$ and of how the coordinates are extended off  the cone.

Throughout this work an overline denotes the restriction of space-time objects to $C_O$.

%while $\sim$ is used to indicate that the corresponding tensor is defined on $S^{n-1}$ and associated to the $r$-dependent metric $\tilde g_{\Sigma_r}$.

The restriction of the inverse metric to the light-cone takes the form
\begin{equation*}
 {\overline g}^\# \equiv \overline g^{\mu\nu} \partial_\mu\partial_\nu= 2\nu^0\partial_u\partial_r + \overline g^{11} \partial_r\partial_r + 2\overline g^{1A}\partial_r\partial_A + \overline g^{AB}\partial_A\partial_B
\;,
\end{equation*}
where
\begin{equation*}
   \nu^0:=\overline g^{01}=(\nu_0)^{-1}\;, \enspace \nu^A:=\overline g^{AB}\nu_B \;,\enspace
\overline g^{1A} =-\nu^0\nu^A\;,\enspace \overline g^{11}=(\nu^0)^2(\nu^A\nu_A - \overline g_{00})
\;,
\end{equation*}
and where $\overline g^{AB}$ is the inverse of $\overline g_{AB}$.
The coordinate transformation relating the two coordinate systems $(y^{\mu})$ and $(x^{\mu})$ takes the form
\begin{equation*}
 u = \hat r - y^0\;, \quad r = \hat r\;, \quad x^A=\mu^A(y^i/\hat r)\;, \quad \text{with} \quad \hat r := \sqrt{\sum_i (y^i)^2 }
 \;.
\end{equation*}
The inverse transformation reads
\begin{equation*}
 y^0=r-u\;, \quad y^i=r\Theta^i(x^A)\;, \quad \text{with}\quad \sum_i(\Theta^i)^2=1
 \;.
\end{equation*}
Adapted null coordinates are singular at the vertex of the cone $C_O$ and $C^{\infty}$ elsewhere.
They are convenient to analyze the initial data constraints satisfied by the trace $\overline g$ on the light-cone.
Note that the space-time metric $g$ will in general not be of the form (\ref{null2}) away from $C_O$.
We further remark that adapted null coordinates are not uniquely fixed, for there remains the freedom to redefine the coordinate $r$ (the only restriction being that $r$ is strictly increasing on the generators and that $r=0$ at the vertex; compare Section~\ref{kappa_freedom} below),
 %x%\ptcr{reworded, as queried by the referee}
 and to choose local coordinates on $S^{n-1}$.

\subsubsection{Generalized wave-map gauge}
\label{generalizedwg}

%x% \ptcr{rewordings throughout the section, as requested by the referee}
Let us be given an auxiliary Lorentzian metric $\hat g$.
A standard method to establish existence, and well-posedness, results for Einstein's vacuum field equations $R_{\mu\nu}=0$ is  a ``hyperbolic reduction'' where the Ricci tensor is replaced
by the \textit{reduced Ricci tensor  in $\hat g$-wave-map gauge},
\begin{equation}
 \label{17V14.1}
 R^{(H)}_{\mu\nu} := R_{\mu\nu} - g_{\sigma(\mu}\hat\nabla_{\nu)}H^{\sigma}
 \;.
\end{equation}
Here
\begin{equation}
 H^{\lambda} :=\Gamma^{\lambda}-\hat \Gamma^{\lambda} - W^{\lambda}
\;, \quad
\Gamma^{\lambda}:= g^{\alpha\beta}\Gamma^{\lambda}_{\alpha\beta}\;, \quad
\hat \Gamma^{\lambda}:= g^{\alpha\beta}\hat \Gamma^{\lambda}_{\alpha\beta}
\;.
\end{equation}
We use the hat symbol ``$\,\hat\enspace\,$" to indicate quantities associated with the \textit{target metric $\hat g$},
while
$W^{\lambda}=W^{\lambda}(x^{\mu},g_{\mu\nu})$  denotes a  vector field  which  is allowed to depend upon the coordinates
and the metric $g$,  but not upon derivatives of $g$.

The \textit{wave-gauge vector $H^{\lambda}$} has been chosen of the above form~\cite{FriedrichCMP,Friedrich:hyperbolicreview,CCM2} to remove  some second-derivatives terms in the Ricci tensor, so that the \emph{reduced vacuum
Einstein equations}
\bel{17V14.3}
 R^{(H)}_{\mu\nu}=0
 \ee
 form a system of quasi-linear wave equations for $g$.

 Any solution of \eq{17V14.3} will provide a solution of the vacuum Einstein equations provided that the so-called \textit{$\hat g$-generalized wave-map gauge condition}
 \bel{17V14.2}
 H^\lambda=0
 \ee
is satisfied. In the context of the characteristic initial value problem, the ``gauge condition'' \eq{17V14.2} is satisfied by solutions of the reduced Einstein equations  if it is  satisfied on the initial characteristic hypersurfaces.

The vector field $W^\lambda$ reflects the freedom to choose coordinates off the cone. Its components can be freely specified, or chosen to satisfy ad hoc equations.
Indeed, by a suitable choice of coordinates the gauge source functions $W^{\lambda}$ can locally be given any preassigned form, and conversely the $W^{\lambda}$'s can be used to determine coordinates by solving wave equations, given appropriate initial data on the cone.

In most of this work we will use a Minkowski target in adapted null coordinates, that is
\begin{equation}
 \hat g  = \eta \equiv -\mathrm{d}u^2 + 2\mathrm{d}u\mathrm{d}r + r^2s_{AB}\mathrm{d}x^A\mathrm{d}x^B
 \;,
 \label{Minktarget}
\end{equation}
where $s$ is the  unit round metric on the sphere $S^{n-1}$.

\subsection{The first constraint equation}
 \label{kappa_freedom}

%x%\tim{added}
Set  $\ell\equiv \ell^{\mu}\partial_
 \mu\equiv \partial_r $.
The Raychaudhuri equation $\overline R_{\mu\nu}\ell^{\mu}\ell^{\nu}\equiv \overline R_{11}=0$ provides a constraining relation between  the connection coefficient $\kappa$ and other geometric objects on $C_O$, as follows: Recall that the \textit{null second fundamental form} of $C_O$ is defined as
\begin{equation*}
\chi_{ij} \,:=\, \frac{1}{2} (\mathcal{L}_{\ell} \coneg)_{ij}
 \;,
\end{equation*}
where $\mathcal{L}$ denotes the Lie derivative. In the adapted coordinates described above we have
\begin{equation*}
\chi_{AB} = -\overline\Gamma{}^0_{AB}\nu_0  = \frac{1}{2}\partial_r\overline g_{AB}
\;, \quad
 \chi_{11}\,=\,0\;, \quad \chi_{1A}\,=\,0
 \;.
\end{equation*}
The null second fundamental form is sometimes called \emph{null extrinsic curvature} of the initial surface $C_O$, which is misleading since
only objects intrinsic to $C_O$ are involved in its definition.

The \textit{mean null extrinsic curvature} of  $C_O$, or the \textit{divergence} of  $C_O$, which we denote by $\tau$
and which is often denoted by $\theta$ in the literature, is defined as the trace of $\chi$:
\begin{equation}
 \tau:= \chi_A^{\phantom{A}A}\equiv \overline g^{AB}\chi_{AB}\equiv \frac{1}{2}\overline g^{AB}\partial_r\overline g_{AB} \equiv \partial_r{\log\sqrt{\det\coneg_{\Sigma_r}}}
 \;.
 \label{definition_tau}
\end{equation}
It measures the rate of change of area along the null geodesic generators of $C_O$.
The traceless part of $\chi$,
\begin{eqnarray}
 \sigma_A^{\phantom{A}B} &:=& \chi_A^{\phantom{A}B} - \frac{1}{n-1}\delta_A^{\phantom{A}B}\tau \,\equiv\, \overline g^{BC}\chi_{AC} - \frac{1}{n-1}\delta_A^{\phantom{A}B}\tau
 \label{definition_sigmaAB}
\\
 &=& \frac{1}{2}\gamma^{BC}(\partial_r\gamma_{AC})\breve{}
 \;,
 \label{formula_sigmaAB}
\end{eqnarray}
is known as the \textit{shear} of $C_O$.
In (\ref{formula_sigmaAB}) the field $\gamma$ is any representative of the conformal class of $\check g_{\Sigma_r}$, which is sometimes regarded as the free initial data.
The addition of the ``$\breve{~~} $''-symbol to a tensor $w_{AB}$ denotes ``the  trace-free part of'':
\begin{equation}
 \label{12XII13.1}
\breve{w}_{AB} :=w_{AB} - \frac{1}{n-1} \gamma_{AB}\gamma^{CD}w_{CD}
\;.
\end{equation}
We set
\begin{eqnarray}
 |\sigma|^2 &:=& \sigma_A^{\phantom{A}B}\sigma_B^{\phantom{B}A} = - \frac{1}{4}(\partial_r\gamma^{AB})\breve{}\,(\partial_r\gamma_{AB})\breve{}
 \label{definition_sigma}
%\\
%&=&- \frac{1}{4}\partial_r\gamma^{AB}\partial_r\gamma_{AB} - \frac{1}{4(n-1)}(\gamma^{AB}\partial_r\gamma_{AB})^2
 \;.
\end{eqnarray}
We thus observe that the shear $\sigma_A{}^B$ depends merely on the conformal class of $\check g_{\Sigma_r}$.
This is not true for $\tau$, which is instead in one-to-one correspondence with the conformal factor relating $\check g_{\Sigma_r}$ and $\gamma$.

Imposing the generalized wave-map gauge condition $H^{\lambda}=0$, the wave-gauge constraint equation  induced by $\ol R_{11}=0$
reads \cite[equation (6.13)]{CCM2},
\begin{equation}
 \partial_r\tau - \underbrace{\Big( \nu^0\partial_r\nu_0 - \frac{1}{2}\nu_0(\overline W{}^0  + \ol {\hat\Gamma}^0) - \frac{1}{2}\tau \Big)}_{=:\kappa}\tau + |\sigma|^2 + \frac{\tau^2}{n-1} = 0
 \;.
\label{R11_constraint}
\end{equation}
%

%x%\ptcr{added, as requested by the referee}
Under the allowed  changes of the coordinate $r$, $r\mapsto \overline r(r,x^A)$, with $\partial \overline r/\partial r>0$, $\overline r(0,x^A)=0$,
the tensor field $g_{AB}$ transforms as a scalar,
\bel{17V14.7}
  \overline g_{AB}(\overline r, x^C)
   =
    g_{AB}(r(\overline r, x^C),x^C)
    \;,
\ee
the field $\kappa$ changes as a connection coefficient
\bel{17V14.5}
 \bar \kappa =  \frac{\partial r}{\partial\overline r}  \kappa +  \frac{\partial \overline r}{\partial r}  \frac{\partial^2 r}{\partial \overline r ^2}
 \;,
\ee
while $\tau$ and $\sigma_{AB}$ transform as one-forms:
\bel{17V14.6}
 \overline \tau = \frac {\partial r}{\partial \overline r} \tau\;,
 \quad
 \overline \sigma_{AB} = \frac {\partial r}{\partial \overline r}  \sigma_{AB}
 \;.
\ee

The freedom to choose $\kappa$ is thus directly related to the freedom to reparameterize the generators of $C_O$.
 Geometrically, $\kappa$ describes the acceleration of the integral curves of $\ell$, as seen from the identity $\nabla_{\ell}\ell^{\mu}=\kappa\ell^{\mu}$.
The choice $\kappa=0$ corresponds to the requirement that the coordinate $r$ be an affine parameter along the rays.
For a given $\kappa$ the first constraint equation splits into an equation for $\tau$ and, once this has been solved, an equation for $\nu_0$.

 Once a parameterization of generators has been chosen,
we see that the metric function $\nu_0$ is largely determined by the choice of the gauge-source function  $\overline W{}^0$ and, in fact,
the remaining gauge-freedom in $\nu_0$ can be encoded in $\overline W{}^0$.

\subsection{The wave-map gauge characteristic constraint equations}

Here we present the whole hierarchical ODE-system
of  Einstein wave-map gauge constraints induced by the vacuum Einstein equations in a generalized wave-map gauge (cf.~\cite{CCM2} for details)
for given initial data   $([\gamma],\kappa)$ and gauge source-functions $\overline W^\lambda$.

The equation \eq{R11_constraint} induced by $\ol R_{11}=0$ leads to the equations
\begin{eqnarray}
 \partial_r\tau - \kappa \tau + |\sigma|^2 + \frac{\tau^2}{n-1} &=& 0
 \;,
 \label{constraint_tau}
\\
 \partial_r\nu^0 + \frac{1}{2}(\overline W{}^0+  \ol {\hat\Gamma}^0) + \nu^0(\frac{1}{2}\tau + \kappa ) &=& 0
 \;.
 \label{constraint_nu0}
\end{eqnarray}
Equation \eq{constraint_tau} is a Riccati differential equation for $\tau$ along each null ray, for $\kappa=0$ it reduces to the standard form of the Raychaudhuri equation.
Equation \eq{constraint_nu0} is expressed in terms of %
$$
 \nu^0:=\frac 1 {\nu_0}
$$
rather than of $\nu_0$, as then it becomes linear.
Our aim is  to analyze the asymptotic behavior of solutions of the constraints, for this
it turns out to be convenient to introduce an auxiliary positive function $\varphi$, defined as
\begin{equation}
 \tau =(n-1)\partial_r\log\varphi
 \;,
\label{relation_tau_phi}
\end{equation}
which transforms \eq{constraint_tau} into a second-order \textit{linear} ODE,
\begin{equation}
 \partial^2_{r}\varphi -\kappa\partial_r\varphi + \frac{|\sigma|^2}{n-1}\varphi =0
 \;.
 \label{constraint_phi}
\end{equation}
The function $\varphi$ is essentially a rewriting of the conformal factor $\Omega$ relating $\coneg $
 and the initial data $\gamma$,
$\overline g_{AB} = \Omega^2 \gamma_{AB}$:
\begin{equation}
 \Omega = \varphi \left( \frac{\det s}{\det \gamma}\right)^{1/(2n-2)}
 \;.
 \label{definition_Omega}
\end{equation}
%
%x%\tim{added}
Here $s=s_{AB}\mathrm{d}x^A\mathrm{d}x^B$ denotes the standard metric on $S^{n-1}$. The initial data symmetric tensor field $\gamma=\gamma_{AB}dx^A dx^B$ is assumed to form a
one-parameter family of Riemannian metrics $r\mapsto \gamma(r,x^A)$ on $S^{n-1}$.

The boundary conditions  at the vertex $O$ of the cone for the ODEs occurring in this work follow from the requirement of regularity of the metric there.
When imposed, they guarantee that (\ref{constraint_nu0}) and (\ref{constraint_phi}), as well as all the remaining constraint equations below, have  unique solutions. The relevant conditions at
the vertex have been computed in regular coordinates and then translated into adapted null coordinates in~\cite{CCM2} for a  natural family of gauges.

For $\nu^0$ and $\varphi$ the boundary conditions read
\begin{eqnarray*}
 \begin{cases}
  \lim_{r\rightarrow 0}\nu^0 = 1
\;,
\\
 \lim_{r\rightarrow 0}\varphi=0\;,\quad \lim_{r\rightarrow 0}\partial_r\varphi=1
 \;.
 \end{cases}
\end{eqnarray*}

The Einstein equations $\overline R_{1A} = 0$ imply the equations
\cite[Equation (9.2)]{CCM2} (compare~\cite[Equation~(3.12)]{ChPaetz})
\begin{eqnarray}
  \frac{1}{2}(\partial_r + \tau)\xi_A  - \conenabla_B \sigma_A^{\phantom{A}B} + \frac{n-2}{n-1}\partial_A\tau +\partial_A \kappa
=0
 \;,
 \label{eqn_nuA_general}
\end{eqnarray}
where $\conenabla$ denotes the Riemannian connection defined by  $\coneg_{\Sigma_r}$,
and
$$
 \xi_A:=-2\ol \Gamma^1_{1A}
 \;.
$$
When  $\ol H^0=0 $  one has $\ol H^A=0$ if and only if
\begin{eqnarray}
 \xi_A
&= & -2\nu^0\partial_r\nu_A + 4\nu^0\nu_B\chi_A{}^B + \nu_A(\overline W{}^0+ \ol {\hat\Gamma}^0) + \overline g_{AB}( \overline W{}^B
+ \ol {\hat\Gamma}^B)
 \nonumber
\\
 &&- \gamma_{AB} \gamma^{CD} \coneGamma{}^B_{CD}
% \label{eqn_xiA_gen}
%\\
% &=& -2\nu^0\partial_r\nu_A + 4\nu^0\nu_B\chi_A^{\phantom{A}B} + \Big(  \overline W{}^0 + \ol {\hat\Gamma}^0-\frac{2}{r}\nu^0\Big)\nu_A + \overline g_{AB}\overline W{}^B
% \nonumber
%\\
% &&+ \gamma_{AB}\gamma^{CD}(\mathring\Gamma^B_{CD} - \coneGamma^B_{CD})
  \;.
 \label{eqn_xiA}
\end{eqnarray}
Here $\coneGamma^B_{CD}$ are the Christoffel symbols associated to the  metric $\coneg_{\Sigma_r}$.

Given fields $\kappa$ and $\overline g_{AB}=g_{AB}|_{u=0}$ satisfying the Raychaudhuri constraint equation,
%x% \ptcr{comment added, as requested by the referee}
 the equations (\ref{eqn_nuA_general}) and \eq{eqn_xiA} can be read as hierarchical linear first-order PDE-system which successively determines
$\xi_A$ and $\nu_A$ by solving ODEs. The boundary conditions at the vertex are
\begin{equation*}
 \lim_{r\rightarrow 0}\nu_A = 0=\lim_{r\rightarrow 0}\xi_A
 \;.
\end{equation*}

The remaining constraint equation follows from the Einstein equation $\overline g^{AB} \overline R_{AB} = 0$
\cite[Equations (10.33) \& (10.36)]{CCM2},
\begin{eqnarray}
 (\partial_r + \tau + \kappa)\zeta +  \coneR - \frac{1}{2}\xi_A\xi^A +\conenabla_A\xi^A =0
 \;,
 \label{zeta_constraint}
\end{eqnarray}
where we have set $\xi^A:= \ol g^{AB}\xi_B$.
The function $\coneR$ is the curvature scalar associated to $\coneg_{\Sigma_r}$.
The auxiliary function $\zeta$ is defined as
\begin{equation}
 \zeta:= (2\partial_r + \tau + 2\kappa)\overline g^{11} + 2\overline W{}^1 + 2 \ol {\hat\Gamma}^1
 \;,
 \label{dfn_zeta}
\end{equation}
and satisfies, if $\ol H^{\lambda}=0$, the relation $\zeta=2\ol g^{AB}\ol\Gamma^1_{AB} + \tau \ol g^{11}$.
The term $\ol{{\hat\Gamma}}{}^1$ depends upon the target metric chosen, and with our current Minkowski target $\hat g=\eta$ we have
\begin{equation}
 \label{5XII13.1}
 \ol {\hat\Gamma}^1 = \ol {\hat\Gamma}^0=  - r\overline g^{AB}s_{AB}
 \;.
\end{equation}
Taking the relation
\begin{equation}
 \overline g^{11} = (\nu^0)^2(\nu^A\nu_A - \overline g_{00})
\end{equation}
into account, the definition \eq{dfn_zeta} of $\zeta$ becomes an equation for $\overline g_{00}$ once $\zeta$ has been determined.
The boundary conditions for \eq{zeta_constraint} and \eq{dfn_zeta} are
\begin{equation*}
 \lim_{r\rightarrow 0}\overline g^{11} =1\;,\quad \lim_{r\rightarrow 0}(\zeta+2 r^{-1}) =0
\;.
\end{equation*}

\input{GlobalSolutions}
\section{Preliminaries to solve the constraints asymptotically}
 \label{s12XII13.2}

\subsection{Notation and terminology}
 \label{ss12XII13.2}

Consider a metric which has a smooth, or polyhomogeneous, conformal completion at infinity \emph{\`a la Penrose}, and suppose that the closure (in the completed space-time) $\ol{\mcN}$ of a null hypersurface  $\mcN$ of $O$  meets $\scrip$ in a smooth sphere.
One can then introduce Bondi coordinates $(u,r,x^A)$ near $\scrip$, with $\ol{\mcN}\cap \scrip$ being the level set of a Bondi retarded coordinate $u$ (see~\cite{TamburinoWinicour} in the smooth case, and \cite[Appendix~B]{ChMS} in the polyhomogeneous case).
The resulting Bondi area coordinate $r$ behaves as $1/\Omega$, where $\Omega$ is the compactifying factor. If one uses $\Omega$ as one of the coordinates near $\scrip$, say $x$, and chooses $1/x$ as a parameter along the generators of $\mcN$, one is led to an asymptotic behaviour of the metric which is captured by the following definition:
%\tim{emphasize that this existence of suitable adapted null coord is nec. for a space-time to admit a smooth conformal completion a la Penrose}

%
\begin{definition}
\label{definition_smooth}
\rm{
We say that a smooth metric tensor $\overline g_{\mu\nu}$  defined on a null hypersurface $\mcN$ given in adapted null coordinates has a
\textit{smooth conformal completion at infinity}
if the unphysical metric tensor field $\overline{\tilde g}_{\mu\nu}$
obtained via
the coordinate transformation $r\mapsto 1/r=: x$ and the conformal rescaling $\overline g\mapsto  \overline{\tilde g} \equiv x^2 \overline g$
is, as a Lorentzian metric, smoothly extendable across $\{x=0\}$. We will say that $\overline g_{\mu\nu}$ is \emph{polyhomogeneous} if the conformal extension obtained as above is polyhomogeneous at $\{x=0\}$, see Appendix~\ref{A22XII13.1}.

The components of a smooth tensor field on $\mcN$ will be said to be \textit{smooth  at infinity},
respectively \emph{polyhomogeneous at infinity}, whenever they admit, in the $(x,x^A)$-coordinates, a smooth, respectively polyhomogeneous, extension in the conformally rescaled space-time across $\{x=0\}$.
}
\end{definition}

 \begin{Remark}
  \label{R12XII13.1}
 {\rm The reader is warned that the definition contains an implicit restriction, that} $\mcN$
 is a smooth hypersurface in the conformally completed space-time. {\rm In the case of a light-cone, this excludes existence of points which are conjugate
 to $O$ both in $\mcM$ and on $\overline C_O\cap \scrip$.}
 \end{Remark}

We emphasise that Definition~\ref{definition_smooth} concerns only fields on $\mcN$, and \emph{no assumptions are made concerning existence, or properties, of an associated space-time.} In particular there \emph{might not} {be} an associated space-time; and if there is one, it \emph{might or might not} have a smooth completion through a conformal boundary at null infinity.

The conditions of the definition are both conditions on the metric and on the coordinate system. While the definition restricts  the class of parameters $r$,   there remains considerable freedom, which  will be exploited in what follows.
It should  be clear that the existence of a coordinate system  as above on a globally-smooth light-cone is a necessary condition for a space-time to admit a smooth conformal completion at null infinity, for points $O$ such that $\overline{C}_O\cap \scrip$ forms a smooth hypersurface in the conformally completed space-time.

Consider a real-valued function
\begin{equation*}
 f : (0,\infty)\times S^{n-1} \longrightarrow \mathbb{R} \;, \quad (r,x^A) \longmapsto f(r,x^A)
 \;.
\end{equation*}
If this function admits an asymptotic expansion in terms of powers of $r$ (whether to finite or arbitrarily high order)
we denote by $f_n$, or $(f)_n$, the coefficient of $r^{-n}$
%\ptcr{corrected}
 in the expansion.

 We will write $f=\mathcal{O}(r^N)$ (or $f=\mathcal{O}(x^{-N})$, $x\equiv1/r$), $N\in\mathbb{Z}$   if
the function
% \ptcr{this is not coherent with the statement that tildes have to do with restrictions to spheres}
%
\begin{equation}
 \label{4XII13.1}
 F(x,\cdot) := x^N f(x^{-1},\cdot )
\end{equation}
is smooth at $x=0$. We emphasize that this is a restriction on $f$ for large $r$, and the condition does not say anything
about the behaviour of $f$ near the vertex of the cone (whenever relevant), where $r$ approaches zero.

We write
 \[   f(r,x^A) \ourdoteq
   \sum_{k=-N}^{\infty} f_k(x^A) r^{-k} \]
if the right-hand side is the asymptotic expansion at $x=0$ of the function $x\mapsto r^{-N} f(r, \cdot)|_{r=1/x}$, compare Appendix~\ref{A22XII13.1}.

The next lemma summarizes some useful properties of the symbol $\mathcal{O}$:
\begin{lemma}
 Let $f=\mathcal{O}(r^N)$ and $g=\mathcal{O}(r^M)$ with $N,M\in\mathbb{Z}$.
\begin{enumerate}
 \item $f$ can be asymptotically expanded as a power series starting from $r^N$,
 \[   f(r,x^A) \ourdoteq
   \sum_{k=-N}^{\infty} f_k(x^A) r^{-k} \]
 for some suitable smooth functions $f_k: S^{n-1} \rightarrow \mathbb{R}$.
 \item The $n$-th order derivative, $n\geq 0$, satisfies
\[ \partial^n_rf(r,x^A) = \begin{cases}
   \mathcal{O}(r^{N-n})
    \;,
   & \text{for $N<0$,} \\
       \mathcal{O}(r^{N-n})
       \;,
   & \text{for $N\geq 0$ and  $N-n\geq 0$,} \\
   \mathcal{O}(r^{N-n-1})
   \;,
    & \text{for  $N\geq 0$ and $N-n\leq -1$,}
  \end{cases}
\]
  as well as
  \[ \partial^n_A f(r,x^B) = \mathcal{O}(r^N)\;. \]
\item $f^n  g^m = \mathcal{O}(r^{nN+mM})$ for all $n,m\in\mathbb{Z}$.
\end{enumerate}
\end{lemma}

\subsection{Some a priori restrictions}
\label{apriorirestrictions}

In order to solve the constraint equations asymptotically and derive necessary-and-sufficient conditions concerning smoothness of the solutions at infinity in adapted coordinates, it is convenient to have some a priori knowledge regarding the lowest admissible orders of certain functions appearing in these equations,
and to exclude the appearance of logarithmic terms in the expansions of  fields such as $\xi_A$ and $\ol W^{\lambda}$.
Let us therefore derive the necessary restrictions on the metric, the gauge source functions, etc.\
needed to end up with a trace of a metric on the light-cone which admits a smooth conformal completion at infinity.

\subsubsection{Non-vanishing of $\varphi$ and $\nu^0$}

As described above, the Einstein wave-map gauge constraints
can be represented as a system of  \textit{linear} ODEs for $\varphi$, $\nu^0$, $\nu_A$ and $\ol{g}^{11}$,
so that existence and
uniqueness (with the described boundary conditions)
 of global solutions is guaranteed if the coefficients in the relevant ODEs are globally defined.
Indeed, we have to make sure that the resulting symmetric tensor field $\overline g_{\mu\nu}$ does not degenerate,  so that it  represents a regular Lorentzian metric in the respective adapted null coordinate system.
In a setting where the starting point are conformal data $\gamma_{AB}(r,\cdot)\mathrm{d}x^A\mathrm{d}x^B$ which
 define a Riemannian metric for all $r>0$, this will be the case if and only if $\varphi$ and $\nu^0$ are nowhere vanishing, in fact \emph{strictly positive} in our conventions,
\begin{equation}
 \varphi >0 \;, \  \nu^0 >0 \quad \forall \,r>0\;.
\end{equation}

\subsubsection{A priori restrictions on  $\overline g_{\mu\nu}$}
\label{apriori_subsection}

Assume that $\ol g_{\mu\nu}$ admits a smooth conformal completion in the sense of Definition~\ref{definition_smooth}. Then
 its conformally rescaled counterpart $\ol{\tilde g}_{\mu\nu} \equiv x^2 \ol g_{\mu\nu}$
satisfies
\begin{equation}
 \overline{\tilde g}_{\mu\nu} = \mathcal{O}(1) \quad \text{with} \quad \ol{\tilde g}_{0x}|_{x=0}> 0\;, \quad \det \ol{\tilde g}_{AB}|_{x=0} > 0
 \;.
\end{equation}
This imposes the following restrictions on the admissible asymptotic form of the components $g_{\mu\nu}$ in  adapted null coordinates $(u,r\equiv 1/x,x^A)$:
\begin{equation}
 \nu_0 \,=\, \mathcal{O}(1)\;, \quad
 \nu_A \,=\,  \mathcal{O}(r^2)\;, \quad
 \overline g_{00} \,=\, \mathcal{O}(r^2)\;, \quad
 \overline g_{AB} \,=\, \mathcal{O}(r^2)
\;,
\label{asympt_rest}
\end{equation}
with
\begin{equation}
 (\nu_0)_0 > 0 \quad \text{and} \quad (\det \coneg_{\Sigma_r})_{-4}> 0
 \;.
\label{det_cond}
\end{equation}
Moreover,
\begin{eqnarray}
\tau \,\equiv\, \frac{1}{2}\overline g{}^{AB}\partial_r \overline g_{AB} \,=\,  \frac{n-1}{r} + \mathcal{O}(r^{-2})\;,
\end{eqnarray}
and (recall that $\tau = (n-1)\partial_r\log\varphi $)
\begin{eqnarray}
%\quad \overset{\eq{relation_tau_phi}}{\Longrightarrow} \quad
\varphi \,=\,\varphi_{-1}r + \mathcal{O}(1) \quad \text{for some positive function $\varphi_{-1} $ on $S^{n-1}$.}
\label{asympt_rest_phi}
\end{eqnarray}
Indeed assuming that $\varphi_{-1}$ vanishes for some $x^A$, the function $\varphi$ does not diverge as $r$ goes to infinity along some null ray $\Upsilon$ emanating from $O$,
i.e.\ $\varphi|_{\Upsilon}= \mathcal{O}(1)$ and $\det \coneg_{\Sigma_r}|_{\Upsilon} \equiv (\varphi^{2(n-1)}\det s) |_{\Upsilon} = \mathcal{O}(1)$,
 which is incompatible with~\eq{det_cond}.

The assumptions $\varphi(r,x^A)>0$ and $\varphi_{-1}(x^A)>0$ imply the non-existence of conjugate points on the light-cone up-to-and-including conformal infinity.

\subsubsection{A priori restrictions on gauge source functions}

Assume that there exists a smooth conformal completion of the  metric,
as in Definition~\ref{definition_smooth}. We wish to find the class of gauge functions
$\kappa$ and $\ol W^
\mu$ which are compatible with this asymptotic behaviour.

The relation $\overline g_{AB}= \mathcal{O}(r^2)$ together with $\partial_r = -x^2\partial_x$  and the definition \eq{definition_sigmaAB} implies
\begin{equation}
 \mbox{  $ \sigma_A{}^B=\mathcal{O}(r^{-2})$,\quad  $|\sigma|^2=\mathcal{O}(r^{-4})$. }
 \label{5XII13.2}
\end{equation}
Using the   estimate (\ref{asympt_rest_phi}) for $\tau$ and the Raychaudhuri equation \eq{constraint_tau} we find
\begin{equation}
    \kappa= \mathcal{O}(r^{-3})
 \;,
\label{a_priori_kappa}
\end{equation}
where  cancellations in both the leading and the next-to-leading terms in  \eq{constraint_tau} have been used.
Then \eq{constraint_nu0}, \eq{asympt_rest}, \eq{asympt_rest_phi} and \eq{a_priori_kappa} imply
\begin{equation}
 \ol W^0 = \mathcal{O}(r^{-1})
\;.
\label{a_priori_W0}
\end{equation}

Similarly to $\kappa = \ol \Gamma^r_{rr}$, $\xi_A$ %and $\zeta$
corresponds to the restriction to $C_O$ of certain  connection coefficients (cf.\ \cite{CCM2,ChPaetz})
\begin{eqnarray*}
 \xi_A = -2\ol \Gamma^r_{rA}  %\;,  \quad \zeta = 2\ol g^{AB}\ol\Gamma^1_{AB} + \tau \ol g^{11}
\;.
\end{eqnarray*}
We will use this equation to determine the asymptotic behaviour of $\xi_A$; the main point is to show that there needs to exist a gauge in which $\xi_A$ has no logarithmic terms. We note that the argument here requires assumptions about the whole space-time metric and some of its derivatives transverse to the characteristic initial surface, rather than on $\og_{AB}$.

A necessary condition for the space-time metric to  be smoothly extendable across $\scri^+$
is that
%x%\tim{added}
 the Christoffel symbols of the unphysical metric $\tilde g$ in coordinates $(u,x\equiv 1/r,x^A)$ are smooth at $\scri^+$,
%\footnote{
%It should be pointed out that the exclusion of log terms in the asymptotic expansion of $\xi_A$ requires assumptions regarding the space-time metric and not just its restriction
%to the initial surface.
%}
in particular
\begin{equation}
 \ol {\tilde\Gamma}{}^x_{xA}=\mathcal{O}(1) \;. %\quad  \ol {\tilde g}^{AB}\ol {\tilde\Gamma}{}^x_{AB}=\mathcal{O}(1) \;,
\label{smooth_Christoffels}
\end{equation}
The formula for the transformation of Christoffel symbols under conformal rescalings of the metric, $\tilde g = \Theta^2 g$, reads
\begin{eqnarray*}
 \tilde\Gamma^{\rho}_{\mu\nu} &=& \Gamma^{\rho}_{\mu\nu} + \frac{1}{\Theta}\left(\delta_{\nu}^{\phantom{\nu}\rho}\partial_{\mu} \Theta + \delta_{\mu}^{\phantom{\mu}\rho}\partial_{\nu} \Theta -g_{\mu\nu}g^{\rho\sigma}\partial_{\sigma} \Theta\right)
 \;,
\end{eqnarray*}
and
shows that \eq{smooth_Christoffels} is equivalent to
\begin{equation}
 \ol\Gamma^x_{xA}= \mathcal{O}(1)
\;, \quad
%\ol g^{AB} \ol \Gamma^x_{AB} = \frac{n-1}{x}(\ol g^{xx} + \ol g^{0x}\ol{\partial_0\Theta})  + \mathcal{O}(x^2)
\text{or} \quad    \ol\Gamma^r_{rA}= \mathcal{O}(1)
\;;
\end{equation}
the second equation is obtained from the first one using the transformation law of the Christoffel symbols under the coordinate transformation $x\mapsto r\equiv 1/x$.
Hence $\xi_A=\mathcal{O}(1)$.  Inspection of the leading-order terms in \eq{eqn_nuA_general} leads now to
\begin{equation}
 \xi_A = \mathcal{O}(r^{-1})
\;.
\end{equation}
One can insert all this into  \eq{eqn_xiA}, viewed as an equation for $ \ol W^A$, to obtain
\begin{eqnarray*}
 \ol W^A =  \mathcal{O}(r^{-1})
\;.
\end{eqnarray*}
We note the formula
$$
 \zeta =\overline{ 2 g^{AB} \Gamma^r_{AB} + \tau g^{rr}}
$$
which allows one to relate $\zeta$ to the Christoffel symbols of $g$, and hence also to those of $\tilde g$.
However, when relating $\ol{\tilde\Gamma}^x_{AB}$ and $\ol\Gamma^r_{AB}$
derivatives of the conformal factor $\Theta$ appear which are transverse to the light-cone and whose expansion is a priori not clear. Therefore this formula
cannot be used to obtain information about $\zeta$ in a direct way, and one has to proceed
differently.
Assuming, from now on, that we are in space-dimension three,
it will be shown in part II of this work that the above a priori restrictions and the constraint equation \eq{zeta_constraint} \textit{imply} that the auxiliary function
$\zeta$ has the asymptotic behaviour
\begin{equation}
 \zeta=\mathcal{O}(r^{-1})
\;.
\end{equation}
It then follows from  \eq{dfn_zeta}  and \eq{asympt_rest} that
\begin{equation}
 \ol W{}^1 = \mathcal{O}(r)
  \;.
\end{equation}
This is our final  condition on the gauge functions.
To summarize, necessary conditions
for existence of both a smooth conformal completion of the metric $\ol g$  and of smooth extensions of
the connection coefficients $\ol \Gamma^1_{1A}$  are
\begin{eqnarray}
& \xi_A=\mathcal{O}(r^{-1})
\;, \quad \ol W{}^0 = \mathcal{O}(r^{-1})\;, \quad \ol W{}^A = \mathcal{O}(r^{-1})
\;. &
\end{eqnarray}
Moreover,
\begin{eqnarray}
\text{if} \enspace   \zeta=\mathcal{O}(r^{-1})\enspace \text{then} \enspace \ol W{}^1 = \mathcal{O}(r)
\;.
\end{eqnarray}

\input{asymptoticExpansions}

%\ptc{all material which was previously here put in a file AsymptotiExpansionsTim.tex, to go in a separate paper, input commented out; you could leave
%things together in the thesis version, after removing duplicate material from the second file (which should be done anyway)}
%\input{AsymptotiExpansionsTim}

\vspace{1.2em}
\noindent {\textbf {Acknowledgments}}
 Supported in part by the  Austrian Science Fund (FWF): P 24170-N16. Parts of this material are based upon work supported by the National Science Foundation under Grant No. 0932078 000, while the authors were in residence at the Mathematical Sciences Research Institute in Berkeley, California, during the fall semester of 2013.

\appendix
\section{Polyhomogeneous functions}
 \label{A22XII13.1}

A function $f$ defined on an open set $\mcU$ with smooth boundary $\partial \mcU =\{x=0\}$ is said to be \emph{polyhomogeneous} at $x=0$  if $f\in
 {C}^\infty (\mcU)$  and if there exist integers $N_i$, real numbers
$n_i, $ and functions $f_{ij}\in  {C}^{\infty}(\ol \mcU)$
such that
\bel{Edef2}\forall m \in \N, \quad \exists N(m)\in \N,  \quad
f-\sum_{i=0}^{N(m)}\sum_{j=0}^{N_i}f_{ij}x^{n_i}\ln^j x \in
 {C}^m (\ol \mcU)
\;.
\ee
We will then write
$$
 f \sim \sum_{i,j}f_{ij}x^{n_i}\ln^j x
 \;.
$$

\bibliographystyle{amsplain}
\bibliography{../references/hip_bib,%
../references/reffile,%
../references/newbiblio,%
../references/newbiblio2,%
../references/chrusciel,%
../references/bibl,%
../references/howard,%
../references/bartnik,%
../references/myGR,%
../references/newbib,%
../references/Energy,%
../references/dp-BAMS,%
../references/prop2,%
../references/besse2,%
../references/netbiblio,%
../references/PDE}

\end{document}

%% file: LeeNewcommand2.tex
\newcommand{\coneg}{\check g}
\newcommand{\coneGamma}{\check \Gamma}
\newcommand{\conenabla}{\check \nabla}
\newcommand{\coneR}{\check R}

\newcommand{\mnote}[1]%{}
{\protect{\stepcounter{mnotecount}}$^{\mbox{\footnotesize
$%\!\!\!\!\!\!\,
\bullet$\themnotecount}}$ \marginpar{%\color{red}%
\raggedright\tiny\em
$\!\!\!\!\!\!\,\bullet$\themnotecount: #1} }

\newcommand{\gamman}{h^{(n)}}

\newcommand{\ourdoteq}{\sim}

\newcounter{mnotecount}[section]

\renewcommand{\themnotecount}{\thesection.\arabic{mnotecount}}

\newtheorem{theorem}{\sc  Theorem\rm}[section]

\newtheorem{Theorem}[theorem]{\sc  Theorem\rm}

\newtheorem{definition}[theorem]{\sc  Definition\rm}

\newtheorem{lemma}[theorem]{\sc Lemma\rm}
\newtheorem{Lemma}[theorem]{\sc Lemma\rm}

\newtheorem{proposition}[theorem]{\sc Proposition\rm}
\newtheorem{Proposition}[theorem]{\sc Proposition\rm}

\newtheorem{Remark}[theorem]{\sc Remark\rm}

\newcommand{\ol}[1]{\overline{#1}{}}

\newcommand{\jlcax}[1]{}
\newcommand{\eean}{\nonumber\end{eqnarray}}

\newcommand{\og}{{\overline{g}}}

\newcommand{\kk}[1]{}%{\mnote{{\bf If we consider the KK case:} #1}}

\newcommand{\beq}{\begin{equation}}

\newcommand{\FS}       %{F_1} %
                  {F}
\newcommand{\HS} %{F_2}
       {H_{\mbox{\scriptsize volume}}}

{\ptc{this should be removed in the oberwolfach version}}%

\newcommand{\eeal}[1]{\label{#1}\end{eqnarray}}
\newcommand{\bed}{\begin{deqarr}}
\newcommand{\eed}{\end{deqarr}}
\newcommand{\bedl}[1]{\begin{deqarr}\label{#1}}
\newcommand{\eedl}[2]{\arrlabel{#1}\label{#2}\end{deqarr}}

\newcommand{\mcU}{{\mycal U}}

\newcommand{\mcN}{{\mycal N}}

\newcommand{\bel}[1]{\begin{equation}\label{#1}}
\newcommand{\bea}{\begin{eqnarray}}
\newcommand{\bean}{\begin{eqnarray}\nonumber}
\newcommand{\beal}[1]{\begin{eqnarray}\label{#1}}
\newcommand{\eea}{\end{eqnarray}}

\newcommand{\Eq}[1]{Equation~\eq{#1}}

\newcommand{\qed}{\hfill $\Box$ \medskip}
\newcommand{\proof}{\noindent {\sc Proof:\ }}
\newcommand{\be}{\begin{equation}}
\newcommand{\eeq}{\end{equation}}
\newcommand{\ee}{\end{equation}}
\newcommand{\beqa}{\begin{eqnarray}}
\newcommand{\eeqa}{\end{eqnarray}}
\newcommand{\beqan}{\begin{eqnarray*}}
\newcommand{\eeqan}{\end{eqnarray*}}
\newcommand{\ba}{\begin{array}}
\newcommand{\ea}{\end{array}}

\newcommand{\hyp}{\mycal S}

\newcommand{\mcM}{{\mycal M}}
\newcommand{\mcD}{{\mycal D}}

\newcommand{\scri}{{\mycal I}}%
\newcommand{\scrip}{\scri^{+}}%

\newcommand{\warn}[1]%{}%{}
{\protect{\stepcounter{mnotecount}}$^{\mbox{\footnotesize
$%\!\!\!\!\!\!\,
\bullet$\themnotecount}}$ \marginpar{%\color{red}%
\raggedright\tiny\em
$\!\!\!\!\!\!\,\bullet$\themnotecount: {\bf Warning:} #1} }

\newcommand{\N}{\mathbb N}

\newcommand{\eq}[1]{(\ref{#1})}

\newcommand{\ptc}[1]{\mnote{{\bf ptc:}#1}}

\newcommand{\beqar}{\begin{deqarr}}
\newcommand{\eeqar}{\end{deqarr}}

\newcommand{\beaa}{\begin{eqnarray*}}
\newcommand{\eeaa}{\end{eqnarray*}}

%% file: GlobalSolutions.tex
\subsection{Global solutions}
 \label{s6XII13.1}
% \ptc{some remarks moved here from elsewhere and commented out, not clear where they should go now}
%Recall that our ultimate goal in this work is to find necessary-and-sufficient conditions on unconstrained initial data on a smooth light-cone so that the resulting conformally rescaled metrics are  smoothly extendable across $\scri^+$.
%This requires the following three ingredients:
%\begin{enumerate}
%\item The constraint equations need to admit a global solution $\ol g$ on  $C_O$. This will be the case if and only if the functions $\varphi$ and $\nu^0$ are of constant sign on $C_O\setminus \{O\}$.
%\item The metric $\ol g$ needs to be smoothly extendable as a Lorentzian metric, which means that the functions $\varphi_{-1}$ and $(\nu^0)_0$
%need to have a sign.
% \ptcr{not clear}
%\item The components of $\ol g$ need to be smoothly extendable across $\scri^+$. For this one has to make sure that their asymptotic expansions
%contain no logarithmic terms.
%\end{enumerate}
%%
%Point 1.\ has been addressed in Section~\ref{s6XII13.1}.
%In this section we will provide conditions which ensure that  2.\ holds.
%Point 3.\ will be analyzed in detail in~\cite{TimAsymptotics}, with the results there presented in Section~\ref{summary1x}.

A prerequisite for obtaining asymptotic expansions is existence of solutions of the constraint equations defined for all $r$.
The question of globally defined data becomes trivial when all  metric components are prescribed on $C_O$: Then the only condition is that $\tau$, as calculated from $\ol g_{AB}$, is strictly positive. Now, as is well-known, and will be rederived shortly in any case, negativity of $\tau$ implies formation of conjugate points in finite affine time, or geodesic incompleteness of the generators. In this work we will only be interested in light-cones  $C_O$ which are globally smooth (except, of course, at the vertex), and extending all the way to conformal infinity. Such cones have complete generators without conjugate points,
and so $\tau$ must remain positive. But then one can solve algebraically the Raychaudhuri equation to globally determine $\kappa$.

We note that the function $\tau$ depends upon the choice of parameterisation of the generators, but its sign does not, hence the above discussion applies regardless of that choice.
%x%\tim{addition}
Recall that we assume that the tip of the cone corresponds to $r \rightarrow 0$ and that the condition that $\kappa=O(r^{-3})$ ensures that an affine parameter along the generators
tends to infinity for $r\rightarrow \infty$, so that the parameterization of $r$ covers the whole cone from $O$ to null infinity.

In some situations it might be convenient
to request that $\kappa$ vanishes, or takes some prescribed value. In this case the Raychaudhuri equation becomes an equation for the function $\varphi$, and the question of its global positivity arises.

Recall that the initial conditions for $\varphi$ at the vertex are $\varphi(0)=0$ and $\partial_r\varphi(0)=1$, and so both $\partial_r \varphi$ and $\varphi$ are  positive near zero.
Now,  (\ref{constraint_phi}) with $\kappa=0$ shows that $\varphi$ is concave
as long as it is non-negative; equivalently, $\partial_r \varphi$ is non-increasing in the region where $\varphi>0$.
An immediate consequence of this is that if $\partial_r\varphi$ becomes negative at some $r_0>0$, then it stays so, with $\varphi$ vanishing for some $r_0<r_1<\infty$, i.e.\ after some finite affine parameter time. We recover the result just mentioned, that negativity of $\partial_r\varphi$ indicates incompleteness, or occurrence of conjugate points, or both. In the first case the solution will not be defined for all affine parameters $r$, in the second $C_O$ will fail to be smooth for $r>r_1$ by standard results on conjugate points. Since the sign of $\partial_r\varphi$ is invariant under orientation-preserving reparameterisations, we conclude that:

\begin{Proposition}
 \label{P12XII13.1}
Globally smooth and null-geodesically-complete light-cones must have $\partial_r\varphi$  positive.
\end{Proposition}

A  set of conditions  guaranteeing global existence of positive solutions of the Raychaudhuri equation, viewed as an equation for $\varphi$, has been given in~\cite[Theorem~7.3]{CCM2}. Here we shall give an alternative simpler criterion, as follows:

 Suppose, first, that $\kappa=0$. Integration of \eq{constraint_phi} gives
%  \ptcr{made more precise, as suggested by the referee}
%
\begin{eqnarray}
  \partial_r\varphi(r,x^A)  &=&1  -\frac{1}{n-1} \int_0^r \big(\varphi|\sigma|^2\big)(\tilde r,x^A)\,\mathrm{d}\tilde r \,\leq\,  1
   \label{6XII13.3}
% \;.
\end{eqnarray}
as long as $\varphi$ remains positive. Since $\varphi(0)=0$, we see that we always have %
$$
 \varphi(r,x^A)\le r
$$
in the region where $\varphi$ is positive, and in that region it holds
\begin{eqnarray*}
  \partial_r\varphi(r,x^A)   &\ge  & 1   -\frac{1}{n-1} \int_0^r \tilde r \,|\sigma (\tilde r,x^A)|^2\mathrm{d}\tilde r
\\
 &\geq & 1 -\frac{1}{n-1} \int_0^{\infty} \tilde r\, |\sigma (\tilde r,x^A)|^2\mathrm{d} \tilde r
\;.
\end{eqnarray*}
This implies that $\varphi$ is strictly increasing if
\begin{equation}
 \int_0^{\infty} r|\sigma|^2 \,\mathrm{d} r \,<\, n-1
 \;.
\label{second_integral}
\end{equation}
Since $\varphi$ is positive for small $r$ it remains positive as long as $\partial_r \varphi$ remains positive,
and so global
positivity of $\varphi$   is guaranteed  whenever \eq{second_integral} holds.

A rather similar analysis applies to the case $\kappa\ne 0$, in which we set
\bel{30IV12.1}
 H(r,x^A) := \int_0 ^r  {\kappa(\tilde r,x^A)}   \mathrm{d}\tilde r
 \;.
\ee
Let
\begin{eqnarray}
  \varphi(r) = \mathring \varphi (s(r))\;, \quad \text{where} \quad s(r):= \int_0^re^{H(\hat r)}\mathrm{d}\hat r
\;,
\label{dfn_mathring_varphi}
\end{eqnarray}
the $x^A$-dependence being implicit.
The function $s(r)$ is   strictly increasing with $s(0)=0$. If we assume that $\kappa$ is continuous in $r$ with $\kappa(0)=0$, defined for all $r$ and, e.g.,
\bel{9XII13.1}
 \int_0^\infty\kappa>-\infty
  \;,
\ee
then $\lim_{r\rightarrow \infty} s(r)=+\infty$,
and thus the function $r\mapsto s(r)$ defines a differentiable bijection from $\mathbb{R}^+$ to itself.
Consequently, a differentiable inverse function $s\mapsto r(s)$ exists, and is smooth if $\kappa$ is.

Expressed in terms of \eq{dfn_mathring_varphi}, \eq{constraint_phi} becomes
\begin{equation}
\partial^2_s\mathring\varphi( s)
 +e^{-2H(r(s))}  \frac{|\sigma|^2(r(s))}{n-1}\mathring\varphi(s) =0
 \;.
\label{constraint_phi_alter}
\end{equation}
A global solution $\varphi>0$ of  \eq{constraint_phi} exists if and only if a global solution  $\mathring\varphi>0$ of
\eq{constraint_phi_alter} exists.
It follows from the considerations above
(note that $\mathring\varphi(s=0)=0$ and $\partial_s\mathring\varphi(s=0)=1$)
that a sufficient condition for global existence of positive solutions of \eq{constraint_phi_alter} is
\begin{eqnarray}
 \lefteqn{\int_0^{\infty}s e^{-2H(r(s))}   |\sigma|^2(r(s)) \,\mathrm{d} s < n-1
 }
 &&
\nonumber
\\
&&
\Longleftrightarrow  \int_0^{\infty}\Big( \int_0^re^{H(\hat r)}\mathrm{d}\hat r \Big) e^{-H(r)}   |\sigma|^2(r) \,\mathrm{d} r < n-1
 \;.
\label{second_integral_alter}
\end{eqnarray}

Consider now the question of positivity of $\nu^0$.
In the $\kappa=0$-wave-map gauge with Minkowski metric as a target
we have (see~\cite[Equation~(4.7)]{ChConeExistence})
\begin{equation}
  \nu^{0}(r,x^A)
  = \frac{\varphi^{-(n-1)/2}(r,x^A)}{2}\int_0^r \Big(\hat r\varphi^{(n-1)/2}\overline{g}{}^{AB}s_{AB}\Big)(\hat r,x^A)\, \mathrm{d}\hat r
  \;.
 \label{11II.1}
\end{equation}
In an $s$-orthonormal coframe  $\theta^{(A)}$,  $\overline{g}{}^{AB}s_{AB}$ is the sum of the diagonal elements  $\overline{g}{}^{(A)(A)}=\ol g^\sharp(\theta^{(A)},\theta^{(A)})$, $A=1,\ldots,n-1$, where $\ol g^\sharp$ the  scalar product on $T^*\Sigma_r$ associated to $\ol g_{AB}\mathrm{d}x^A \mathrm{d}x^B$,   each of which is positive in Riemannian signature. Hence
$$
 \overline{g}{}^{AB}s_{AB}>0
 \;.
$$
%.
So, for globally positive  $\varphi$ we obtain a globally defined strictly positive $\nu^0$, hence also a globally defined strictly positive $\nu_0\equiv 1/\nu^0$.

When $\kappa\ne 0$,
and
allowing further a non-vanishing $W^{0}$, we find instead
\bean
  \nu^{0}(r,x^A)
  & = &  \frac{\left(e^{-H  }\varphi^{-(n-1)/2}\right)(r,x^A)}{2}
   \times
\\
 &&   \int_0^r
   \Big( e^{H }\varphi^{(n-1)/2}(\hat r\overline{g}{}^{AB}s_{AB} - \ol W{}^0)
\Big)(\hat r,x^A)\, \mathrm{d}\hat r
  \;,
 \label{11II.4v2}
\eea
with $H$ as in \eq{30IV12.1}.
If $\ol W{}^0=0$ we obtain positivity as before.
More generally, we see that a necessary-and-sufficient condition for positivity of $\nu^0$ is positivity of the integral in the last line of \eq{11II.4v2} for all $r$. This will certainly be the case if  the gauge-source function $\ol W{}^0$ satisfies
\begin{equation}
\ol W{}^0<r\overline{g}{}^{AB}s_{AB}= r\varphi^{-2} \Big( \frac{\det\gamma}{\det s} \Big)^{1/(n-1)}\gamma^{AB} s_{AB}
\;.
\label{cond_non-vanish_W0}
\end{equation}

Summarising we have proved:

\begin{Proposition}
 \label{P6XII13.1}
\begin{enumerate}
 \item Solutions of the Raychaudhuri equation with prescribed $\kappa$ and $\sigma$ are global  when \eq{9XII13.1} and \eq{second_integral_alter} hold, and lead to  globally positive functions $\varphi$ and $\tau$.
 \item Any global solution of the Raychaudhuri equation with $\varphi>0$ leads to a globally defined positive  function $\nu_0$ when the gauge source function $\ol W^{0}$ satisfies \eq{cond_non-vanish_W0}. This condition will be satisfied for any $\ol W^{0}\le 0$.
\end{enumerate}
\end{Proposition}

%\ptcr{some previous material moved to recycling2.tex}

\subsection{Positivity of $\varphi_{-1}$ and $(\nu^0)_0$}
\label{nonvanishing_subsection}

For reasons that will become clear in Section~\ref{s12XII13.2}, we are interested in fields $\varphi$ and $\nu_0$ which, for large $r$, take the form
\bel{13XII13.11}
 \varphi(r,x^A)  = {\varphi_{-1}}(x^A)r %+{\varphi_{0}}(x^A)
 +o(r)\;,
 \quad
 \nu^0(r,x^A)  = ( {\nu^0})_0(x^A)  +o(1)\;,
\ee
with $ {\varphi_{-1}}$ and   $ ( {\nu^0})_0$ positive. The object of this section is to provide conditions which guarantee existence of such expansions, assuming a global positive solution $\varphi$.

Let us further assume
that $ e^{-2H }\varphi|\sigma|^2$ is continuous in $r$ with
$$\int_0^\infty  \big(e^{-2H }\varphi|\sigma|^2\big)\big|_{r=r(s)} \mathrm{d}s
=\int_0^\infty  e^{- H }\varphi|\sigma|^2 \mathrm{d}r
<\infty
 \;.
$$
%.
Integration of \eq{constraint_phi_alter} and de l'Hospital rule at infinity give
\begin{eqnarray}
  \mathring \varphi_{-1}:=\lim_{s\to\infty}  \frac{\mathring \varphi(s)}{s} =\lim_{s\to\infty} \partial_s \mathring \varphi(s)  =1  -\frac{1}{n-1} \int_0^\infty e^{- H }\varphi|\sigma|^2\mathrm{d} r
   \label{6XII13.4}
 \;.
\end{eqnarray}
This will be strictly positive if e.g.\ \eq{second_integral_alter} holds, as
\begin{eqnarray*}
 \int_0^r e^{H(\tilde r)}\mathrm{d}\tilde r -\varphi(r) &=&  \int_0^r (e^{H(\tilde r)}-\partial_{\tilde r} \varphi(\tilde r)) \mathrm{d}\tilde r
\\
&=&  \int_0^{r(s)} \underbrace{(1-\partial_{\tilde s}\mathring \varphi(\tilde s)) }_{\geq 0 \text{ by } \eq{6XII13.3}}\mathrm{d} \tilde s \,\geq\, 0
\;,
\end{eqnarray*}
and thus by \eq{6XII13.4} and \eq{second_integral_alter}
\begin{eqnarray*}
\mathring\varphi_{-1} \,\geq\, 1- \frac{1}{n-1}  \int_0^{\infty}\Big( \int_0^re^{H(\hat r)}\mathrm{d}\hat r \Big) e^{-H(r)}   |\sigma|^2(r) \,\mathrm{d} r\,>\, 0
\;.
\end{eqnarray*}

One can now use \eq{6XII13.4} to obtain  \eq{13XII13.11} if we assume that the integral of $\kappa$
over $r$ converges: 
\bel{6XII12.6}
\forall x^A \qquad-\infty < \beta(x^A):=\int_0^{\infty} \kappa(r,x^A)\mathrm{d}r <\infty
\;,
\ee
so that
\be
 \int_0^r \kappa(s,\cdot) \mathrm{d}s = \beta(\cdot) + o (1) % +  O(r^{-\epsilon})
 \;.
\ee
Indeed, it follows from \eq{6XII12.6} that there exists  a constant $C$ such that the parameter $s$ defined in \eq{dfn_mathring_varphi} satisfies
\bel{12XII13.2}
 C^{-1}   \le \frac{\partial s}{\partial r} \le C
 \;,
 \quad
 C^{-1} r \le s \le C r
 \;,
 \quad
 \lim_{r\to\infty} \frac{\partial s}{\partial r} = e^\beta
 \;.
\ee
We then have
\begin{eqnarray}
 \nonumber
  \varphi_{-1} &= &
  \lim_{r\to\infty}  \frac{  \varphi(r)}{r} =\lim_{s\to\infty}  \frac{\mathring \varphi(s )}{r(s)}=
  \lim_{s\to\infty}  \frac{ \partial_s \mathring \varphi(s)}{\partial_s r(s)} = e^{-\beta}\mathring\varphi_{-1}
\\
 &
   =
    &
     e^{-\beta}\bigg(1  -\frac{1}{n-1} \int_0^\infty e^{-H }\varphi|\sigma|^2\mathrm{d} r \bigg)
   \label{6XII13.4xx}
 \;.
\end{eqnarray}

We have proved:
\begin{proposition}
 Suppose that \eq{9XII13.1}, \eq{second_integral_alter} and \eq{6XII12.6} hold.
  Then the function $\varphi$ is globally positive  with $\varphi_{-1}>0$.
\end{proposition}

Consider, next, the asymptotic behaviour of $\nu_0$. In addition to \eq{6XII12.6}, we assume now that $\varphi = \varphi_{-1}r + o(r)$, for some function of the angles $\varphi_{-1}$, and that there exists a bounded
function  of the angles $\alpha$  such that
\beal{6XII12.5}
 &
  \displaystyle
  r\overline g^{AB} s_{AB}-\ol  W{}^0 = \frac \alpha {r} + o (r^{-1}) %+ O(r^{-2-\epsilon})
 \;.
  &
\eea
Passing to the limit $r\to\infty$ in \eq{11II.4v2} one obtains
\bean
  \nu^{0}(r,x^A)
  & = &  \frac {\alpha(x^A)} {n-1} + o (1)
   \;.
 \label{11II.4v3}
\eea
We see thus that
\bean
  (\nu^{0})_0 > 0
  \quad
  \Longleftrightarrow
  \quad
  \alpha > 0
   \;,
   \qquad
  (\nu_{0})_0 > 0
  \quad
  \Longleftrightarrow
  \quad
  \alpha < \infty
   \;.
 \label{11II.4v4}
\eea

\begin{Remark}
 \label{R13XII13.1}
  {\rm
Note that  \eq{6XII12.6} and \eq{6XII12.5}  will hold with  smooth  functions $\alpha$ and $\beta$ when the a priori restrictions \eq{5XII13.2}-\eq{a_priori_W0}, discussed below, are satisfied and when both $\varphi$ and $\varphi_{-1} $ are positive. Recall also that if $\ol W{}^0 \le 0$ (in particular, if $\ol W{}^0 \equiv 0$), then the condition  $\alpha\ge0$ follows from the fact that both $s_{AB}$ and $\ol g_{AB}$ are Riemannian.
}
\end{Remark}

So far we have justified the expansion \eq{13XII13.11}. For the purposes of Section~\ref{s9XII13.1} we need to push the expansion one order further. This is the contents of the following:

\begin{Proposition}
 \label{P6XII13.2}
Suppose that there exists a Riemannian metric $(\gamma_{AB})_{-2} \equiv (\gamma_{AB})_{-2}(x^C)$ and a tensor field $(\gamma_{AB})_{-1} \equiv (\gamma_{AB})_{-1}(x^C)$ on $S^{n-1}$ such that for large $r$ we have
\begin{eqnarray}
 \label{12XII13.11}
 &
  \gamma_{AB} =r^2 (\gamma_{AB})_{-2} + r(\gamma_{AB})_{-1}  + o (r ) \;,
  &
\\
 &
  \partial_r \big(\gamma_{AB}-r^2(\gamma_{AB})_{-2} - r(\gamma_{AB})_{-1} \big)=   o (1 )
  \;,
  &
 \label{12XII13.12}
\\
 \label{9XII13.1x}
 &
 \displaystyle
 \int_0^ r \kappa(s,x^A)\mathrm{d}s = \beta_0(x^ A)+ \beta_{1}(x^A) r^{-2} + o ( r^{-2})
  \;.
  &
\eea
%
%  \ptcr{conditions added, rewordings; how much do I need to control sigma, do I need to get $\varphi_0$ for the argument here? I need two terms in
%  the no-go theorem to have the ln term show up }
Assume moreover that $\varphi$  exists  for all $r$, with $\varphi>0$. Then:
\begin{enumerate}
   \item
There exist  bounded functions of the angles $\varphi_{-1}\ge 0$ and $\varphi_{0}$ such that
\bel{12XII13.21}
    \varphi (r) =   \varphi_{-1} r +   \varphi_{0}   + O(r^{-1})\;.
\ee

\item
If, in addition, $\nu_0$ exists for all $r$, if it holds that $\varphi_{-1}>0$
 and if  $\ol W{}^0$
 takes the form
$\ol W{}^0(r,x^A)=(\ol W{}^0)_1 (x^A) r^{-1} + o(r^{-1})$ with
\bel{9XII13.3}
 (\ol W{}^0)_1< s_{AB} (\ol g^{AB})_{2}=(\varphi_{-1})^{-2} \Big( \frac{\det\gamma_{-2}}{\det s} \Big)^{1/(n-1)}\gamma_{-2}^{AB} s_{AB}
  \;,
\ee
then
%$\nu_0$ is globally defined, positive, and
%
$$
 0
    <(\nu_0)_0<\infty
 \;.
$$
\end{enumerate}
%.
\end{Proposition}

\begin{Remark}
 \label{R18XII13.1}
{\rm
If the space-time is not vacuum, then \eq{constraint_phi_alter} becomes
\begin{equation}
\partial^2_s\mathring\varphi( s)
 +e^{-2H(r(s))}  \frac{\big(|\sigma|^2+\ol R_{rr}\big)(r(s))}{n-1}\mathring\varphi(s) =0
 \;.
\label{constraint_phi_alter nonvac}
\end{equation}
and the conclusions of Proposition~\ref{P6XII13.2} remain unchanged if we assume in addition that
\bel{18XII13.10}
 \ol R_{rr} = O(r^{-4})
 \;.
\ee
}
\end{Remark}

\proof
From \eq{definition_sigma} one finds
$$|\sigma|^2 = O(r^{-4})
\;.
$$
We have already seen that
$$
 \mathring \varphi = \mathring \varphi_{-1}s + o(s)
 \;.
$$
Plugging this in the second term in \eq{constraint_phi_alter} and integrating shows that
$$
 \partial_s \mathring \varphi (s) = \mathring \varphi_{-1}  + O(s^{-2})\;,
 \quad
   \mathring \varphi (s) = \mathring \varphi_{-1} s + \mathring \varphi_{0}   + O(s^{-1})\;.
$$
A simple analysis of the equation relating $r$ with $s$ gives now
$$
 \partial_r   \varphi (r) =  \varphi_{-1}  + O(r^{-2})\;,
 \quad
    \varphi (r) =   \varphi_{-1} r +   \varphi_{0}   + O(r^{-1})\;.
$$
This establishes point 1.

When $\varphi_{-1}$ is positive one finds that \eq{6XII12.5} holds, and from what has been said the result follows.
\qed

\input{NoGoSection}

%% file: NoGoSection.tex
\section{A no-go theorem for  the ($ {\kappa=0}$, $ \ol W{}^0= 0$)-wave-map gauge}
    \label{s9XII13.1}

Rendall's proposal, to solve the characteristic Cauchy problem using the ($ {\kappa=0}$, $ \ol W{}^\mu= 0$)-wave-map gauge, has been adopted by many authors. The object of this section is to show that,
%x%\tim{added}
 in $3+1$ dimensions,
 this approach will always lead to logarithmic terms in an asymptotic expansion of the metric \emph{except for the Minkowski metric}. This makes clear the need to allow non-vanishing gauge-source functions $\ol W^{\mu}$.

% \ptc{this suggests very strongly that $\tau_2$ should be related to the Bondi mass, we need to settle this; note that $\tau_2$ also shows up in
% the asymptotics of $\varphi_0$, where else can it be moved by making coordinate transformations?}
%
More precisely, we prove  (compare~\cite{ChoquetBruhat73}):

\begin{Theorem}
 \label{T9XII13.11}
Consider a four-dimensional vacuum space-time $(\mcM,g)$
which has a
 % \ptcr{relaxed to $C^3$ }
  conformal completion at future null infinity $(\mcM\cup\scrip,\tilde g)$
with a $C^3$ conformally rescaled metric, and suppose that there exists a point $O\in \mcM$ such that
%x%\tim{changed }
$ \overline{C}_O\setminus \{O\}$, where $\overline C_O$ denotes the closure of $C_O$ in $\mcM\cup\scrip$, is a smooth hypersurface in the conformally completed space-time.
If the metric $g$ has no logarithmic terms in its asymptotic expansion for large $r$ in the $  \ol W{}^0=0$  wave-map gauge,
where $r$ is an affine parameter on the generators of $C_O$, then $(\mcM,g)$ is the Minkowski space-time.
\end{Theorem}

\proof
Let $S \subset \scrip$
denote the intersection of $\overline C_O$ with $ \scrip$. Elementary arguments show
%x%\tim{reworded}
that  $\overline C_O$ intersects $ \scrip$ transversally and that $S$ is diffeomorphic to $S^2$. Introduce near $S$ coordinates so that $S$ is given  by the equation $\{u=0=x\}$, where $x$ is an $\tilde g$-affine
parameter along the generators of $\overline C_O$, with $x=0$ at $S$, while the $x^A$'s are coordinates on $S$ in which the metric induced by $\check g$
is manifestly conformal to the round-unit metric $s_{AB}\mathrm{d}x^A \mathrm{d}x^B$ on $S^2$.
(Note that for finitely-differentiable metrics this construction might lead to the loss of one derivative of the metric.) The usual calculation shows that the $g$-affine parameter $r$ along the generators of $\overline C_O$ equals $a(x^A)/x $
%x% \ptcr{we believe that the power is correct as is}
for some positive function of the angles $a(x^A)$. Discarding strictly positive conformal factors, we conclude that for large $r$ the tensor field $\check g$ is conformal to a tensor field $
\gamma_{AB} \mathrm{d}x^A \mathrm{d}x^B$ satisfying
\begin{eqnarray}
 \label{9XII13.21}
 &
 \gamma_{AB} = r^2\big(s_{AB}  + (\gamma_{AB})_{-1} r^{-1} + o (r^{-1} )\big)
 \;,
  &
\\
 &
  \partial_r \big(\gamma_{AB}-r^2 s_{AB} - r(\gamma_{AB})_{-1} \big)=   o (1 )
  \;.
  &
 \label{9XII13.21b}
\eea
The result follows now immediately from~\cite{CCG} and from our next Theorem~\ref{T21IV11.1}.
\qed

\begin{theorem}
 \label{T21IV11.1}
 Suppose that the space-dimension $n$ equals three.
  %x%\ptcr{added}
Let $r|\sigma|$, $r\ol W^0$  and $r^2\ol R_{\mu\nu} \ell^\mu \ell ^\nu$ be bounded for small $r$.
Suppose that $\gamma_{AB}(r,x^A)$ is positive definite for all $r>0$  and admits the expansion \eq{9XII13.21}-\eq{9XII13.21b}, for large $r$
% \ptcr{hypothesis on the derivatives and Ricci added, leading order spherical now }
with the coefficients in the expansion depending only upon $x^C$.
Assume that the first constraint  equation \eq{constraint_phi} with $\kappa=0$ and
$$
 0\le \overline R_{\mu\nu}\ell^\mu \ell^\nu =O(r^{-4})
$$
has a globally defined positive solution satisfying $\varphi(0)=0$, $\partial_r\varphi(0)=1$, $\varphi >0$, and $\varphi_{-1}>0$.
Then there are no logarithmic terms in the asymptotic expansion of $\nu^0$ in a gauge where $ {\kappa=0}$ and $ \ol W{}^0= o(r^{-2})$ (for large $r$)  if and only if
$$
 \sigma\equiv 0 \equiv \overline R_{\mu\nu}\ell^\mu \ell^\nu
 \;.
$$
\end{theorem}

{\noindent \sc Proof of Theorem~\ref{T21IV11.1}:}
At the heart  of the proof lies the following observation:

\begin{Lemma}
  \label{L9XII13.1}
  In space-dimension $n$,
 %x%  \ptcr{higher dimensions allowed, proof extended accordingly}
suppose that $\kappa=0$ and set
\bel{9XII13.9}
 \Psi = r^2 \exp(  \int_0^r \big(\frac {\tau+\tau_1} 2 - \frac {n-1} r\big)\,\mathrm{d}r )
 \;.
\ee
with $\tau_1 \equiv (n-1)/r$
We have $\tau =  (n-1) r^{-1}  + \tau_2 r^{-2} + o(r^{-2})$,
where
\bel{9XII13.10}
 \tau_2 := - \lim_{r\to\infty} r^2 {\Psi^{-1}} \times \int_0^r  (  |\sigma|^2 + \overline  R_{\mu\nu}\ell^\mu \ell^\nu)\Psi \,\mathrm{d}r
 \;,
\ee
provided that the limit exists.
\end{Lemma}

\begin{proof}
Let  $\delta \tau = \tau - \tau_1$.
It follows from the Raychaudhuri equation with $\kappa=0$ that $\delta \tau$ satisfies the equation
$$
 \frac{  \mathrm{d}\delta \tau } {\mathrm{d}r} + \frac {\tau+\tau_1} 2 \delta \tau = -|\sigma|^2 - 8 \pi \ol T_{rr}
 \;.
$$
Solving, one  finds
\beaa
 \delta \tau  & = & -\Psi^{-1}   \int_0^r (  |\sigma|^2 +  8 \pi \ol T_{rr} )\Psi \,\mathrm{d}r
\\
 & = & \frac{\tau_2 }{r^2} +o(r^{-2})
 \;,
\eeaa
as claimed.
\qed
\end{proof}

Let us return to the proof of Theorem~\ref{T21IV11.1}.  Proposition~\ref{P6XII13.2} and Remark~\ref{R18XII13.1} show that
\beal{9XII13.12}
&
 \varphi (r,x^A)= \varphi_{-1}(x^A)r + \varphi_0(x^A) + o(r^{-1})
 \;,
 &
\\
 &
 \tau \,\equiv\, 2\partial_r\log\varphi \,=\,  2r^{-1} -2\varphi_0(\varphi_{-1})^{-1}r^{-2}
 + o(r^{-2})
 %+ \Big(2\frac{\varphi_0^2}{(\varphi_{-1})^2} +\sigma_4  \Big)r^{-3} + \mathcal{O}(r^{-4})
 \;.
 &
\eeal{9XII13.13}
Recall, next, the solution formula \eq{11II.1} for the constraint equation (\ref{constraint_nu0}) with $\kappa=0$  and   $n=3$:
\begin{equation}
  \nu^{0}(r,x^A)
  = \frac{1}{2\varphi (r,x^A)}\int_0^r \varphi \left(s\overline{g}{}^{AB}s_{AB} - \ol W{}^0\right)(s,x^A)\, \mathrm{d}s
  \;.
 \label{11II.1b}
\end{equation}
%
%
%
%%
%\bean
%  \nu^{0}(r,x^A)
%  & = &  \frac{\left(e^{-H  }\varphi^{-(n-1)/2}\right)(r,x^A)}{2}
%   \times
%\\
% &&   \int_0^r
%   \left( e^{H }\varphi^{(n-1)/2}(s\overline{g}{}^{AB}s_{AB} - \ol W{}^0)
%\right)(s,x^A)\, ds
%  \;,
% \label{11II.4v2x}
%\eea
%%
%with $H$ as in \eq{30IV12.1}.
From \eq{9XII13.12}-\eq{9XII13.13} one finds
\begin{equation}
 \ol g^{AB} = r^{-2} (\varphi_{-1})^{-2}[s^{AB} + r^{-1}(\tau_2 s^{AB} - \breve \gamma_{-1}^{AB}) + o(r^{-1})]
\;,
\end{equation}
with
$$
 \breve \gamma_{-1}^{AB} := s^{AC} s^{BD}[(\gamma_{CD})_{-1}- \frac{1}{2}s_{CD}s^{EF}(\gamma_{EF})_{-1} ]\;.
$$
Inserting this into \eq{11II.1b},  and assuming that $\ol W^{0}=o(r^{-2})$, one finds for large $r$
\begin{equation}
 \nu^ 0 =  (\varphi_{-1})^{-2}+  \frac{1}{2}\tau_2 (\varphi_{-1})^{-2}\frac{\ln r} r + O(r^{-1})
 \;,
\end{equation}
with the coefficient of the logarithmic term vanishing if and only if $\tau_2=0$ when a bounded positive coefficient $\varphi_{-1}$ exists.
One can check  that the hypotheses of Lemma~\ref{L9XII13.1} are satisfied, and the result follows.
\qed

%% file: asymptoticExpansions.tex
\section{Asymptotic expansions}
 \label{summary1x}

We have seen in Section~\ref{apriori_subsection} that existence of a smooth completion at null infinity requires $g_{AB}=\mathcal{O}(r^2)$   with $(\det \ol g_{AB})_{-4} > 0$, and thus $\varphi=\mathcal{O}(r)$ with $ \varphi_{-1}> 0$.
But then
\begin{equation*}
 \frac{1}{\sqrt{\det\gamma}}\gamma_{AB} = \varphi^{-2} \frac{1}{\sqrt{\det s}} g_{AB} = \mathcal{O}(1)
 \;.
\end{equation*}
Since only the conformal class of $\gamma_{AB}$ matters, we see that there is no loss of generality to assume that
$\gamma_{AB} = \mathcal{O}(r^2)$, with $(\det \gamma_{AB})_{-4}\ne 0$;
 this is   convenient because then $\gamma_{AB}$ and $\ol g_{AB}$ will display similar asymptotic behaviour.
Moreover, since any Riemannian metric on the 2-sphere is conformal to the standard metric $s=s_{AB}\mathrm{d}x^A\mathrm{d}x^B$,
in the case of smooth conformal completions we may without loss of generality require the initial data $\gamma$ to be   of
the form, for large $r$,
\begin{equation}
 \gamma_{AB} \ourdoteq  r^2 \Big( s_{AB}+ \sum_{n=1}^{\infty} \gamman _{AB} r^{-n}\Big)
 \;,
 \label{initial_data}
\end{equation}
for some smooth tensor fields $ \gamman _{AB}$ on $S^2$.
(Recall that the symbol ``$\ourdoteq $'' has been defined in Section~\ref{ss12XII13.2}.)
If the initial data $\gamma_{AB}$ are not directly of the form \eq{initial_data}, they can either be brought to  \eq{initial_data}   via an appropriate choice of coordinates and conformal rescaling, or they lead to a metric $\ol g_{\mu\nu}$ which is not smoothly extendable across~$\scri^+$.

In the second part of this work~\cite{TimAsymptotics} the following theorem will be proved:
\begin{theorem}
\label{thm_asympt_exp}
Consider the characteristic initial value problem for Einstein's vacuum field equations in four space-time dimensions with smooth
conformal data $\gamma=\gamma_{AB}\mathrm{d}x^A\mathrm{d}x^B$ and gauge  functions $\kappa$ and $\ol W^{\lambda}$ on a cone $C_O$ which has smooth closure
in the conformally completed space-time.
The following conditions are necessary-and-sufficient for the trace of the metric $g=g_{\mu\nu}\mathrm{d}x^{\mu}\mathrm{d}x^{\nu}$ on $C_O$,
obtained as  solution to Einstein's wave-map characteristic vacuum constraint equations \eq{constraint_phi} and \eq{eqn_nuA_general}-\eq{dfn_zeta},
 to admit a smooth conformal completion at infinity and for the connection coefficients $\ol \Gamma^r_{rA}$  to be smooth at $\scri^+$, in the sense of  Definition~\ref{definition_smooth}, when imposing a generalized wave-map gauge condition $H^{\lambda}=0$:
\begin{enumerate}
\item[(i)]  There exists a conformal factor so that the conformally rescaled $\gamma$  satisfies \eq{initial_data}.
\item[(ii)]The functions $\varphi$, $\nu^0$, $\varphi_{-1}$ and $(\nu_0)_0$ have no zeros on $C_O\setminus \{0\}$ and $S^2$, respectively, with the non-vanishing of $(\nu^0)_0$ being equivalent to
\begin{eqnarray}
   (\ol W{}^0)_1
    &< &  %(n-1)
 2(\varphi_{-1})^{-2}
 \;.
\end{eqnarray}

\item[(iii)]
    The gauge  functions satisfy  $\kappa=\mathcal{O}(r^{-3})$,
  $\overline W{}^0=\mathcal{O}(r^{-1})$, $\overline W{}^A=\mathcal{O}(r^{-1})$, $\overline W{}^1=\mathcal{O}(r)$ and, setting $\ol W_A:= \ol g_{AB}\ol W{}^A$,
\begin{eqnarray}
 \label{22XII13.11x}
 (\overline W{}^0)_2
  & = &
    \Big[\frac{1}{2} (\overline W{}^0)_1 + (\varphi_{-1})^{-2}\Big]\tau_2
 \;,
 \label{0gaugecond}
\\
 (\overline W_A)_1 &=& 4(\sigma_A{}^B)_2\mathring\nabla_A\log\varphi_{-1}
   - (\check\varphi_{-1})^{-2}[(\nu_0)_2(\overline W_A)_{-1}
+ (\nu_0)_1(\overline W_A)_0]
 \nonumber
\\
&& - \mathring\nabla_A \tau_2
 - \frac{1}{2}( w_A{}^B)_1( w_B{}^C)_1 (\overline W_C)_{-1}
 -\frac{1}{2} ( w_A{}^B)_2 (\overline W_B)_{-1}
 \nonumber
\\
 &&  - ( w_A{}^B)_1\Big[(\overline W_B)_{0}  +  (\check\varphi_{-1})^{ 2} (\nu_0)_1 (\overline W_B)_{-1}  \Big]
\label{nuA_cond}
 \;,
\\ (\overline W{}^1)_2 &=& \frac{\zeta_2}{2} + (\varphi_{-1})^{-2}\tau_2 + \frac{\tau_2}{4}\coneR_2 + \frac{\tau_2}{2} (\overline W{}^1)_1
 +   \Big[ \frac{\tau_3}{4} + \frac{\kappa_3}{2} - \frac{(\tau_2)^2}{8} \Big](\overline W{}^1)_0
 \nonumber
\\
 &&  \Big[ \frac{1}{48}(\tau_2)^3 - \frac{1}{8}\tau_2\tau_3 - \frac{1}{4}\tau_2\kappa_3 + \frac{1}{6}\tau_4 + \frac{1}{3}\kappa_4 \Big](\overline W{}^1)_{-1}
 \;,
 \label{g00_cond}
\end{eqnarray}
%
 %\ptcr{give a ref or define $\check \varphi$ also trace free part }
where $\mathring\nabla$ is the covariant derivative operator of the unit round metric on the sphere $s_{AB}\mathrm{d}x^A\mathrm{d}x^B$, $\check R_2$ is the $r^{-2}$-coefficient of the scalar curvature $\check R$ of the metric $\check g_{AB}\mathrm{d}x^A\mathrm{d}x^B$,
$\check\varphi_{-1} :=[(\varphi_{-1})^{-2} - \frac{1}{2}(\ol W{}^0)_1]^{-1/2}$,
 and the expansion coefficients  $( w_A{}^B)_n$ are defined using
\begin{equation*}
  w_A{}^B := \Big[ \frac{r}{2}\nu_0(\ol W{}^0 + \ol{\hat\Gamma}{}^0)-1 \Big]\delta_A{}^B +2r\chi_A{}^B
\;.
\end{equation*}
\item[(iv)]
 The  \underline{no-logs-condition}
is satisfied:
  \begin{eqnarray}
  (\sigma_A{}^B)_3  =  \tau_2 (\sigma_A{}^B)_2
%  4\varphi_0(\varphi_{-1})^{-1}[\conenabla_B\sigma_A{}^B]_2  &=& \mathring\nabla_A\sigma_4- 2[\conenabla_B\sigma_A{}^B]_3
 \;.
 \label{no-log-conditions}
 \end{eqnarray}
\end{enumerate}
%The above conditions moreover guarantee that $\ol g^{AB}\ol\Gamma^0_{AB}$
%will also be smoothly extendable across $\scrip$.
%
\end{theorem}

\begin{Remark}
{\rm
If any of the equations \eq{22XII13.11x}-\eq{no-log-conditions} fail to hold, the resulting characteristic initial data sets will have a \emph{polyhomogeneous} expansion in terms of powers of $r$.
}
\end{Remark}

\begin{Remark}
{\rm
Theorem~\ref{thm_asympt_exp} is independent of the particular setting used (and remains also valid when the light-cone is replaced by one of two transversally intersecting null hypersurfaces meeting $\scrip$ in a sphere), cf.\ Section~\ref{s16XII13.1}:
As long as the generalized wave-map gauge condition is imposed one can always compute  $\ol W^{\lambda}$, $\tau$, $\sigma$ etc.\
and check the validity of \eq{0gaugecond}-\eq{no-log-conditions}, whatever the prescribed initial data sets are.
Some care is needed when the Minkowski target is replaced by some other target metric, cf.\ \cite{TimAsymptotics}.
}
\end{Remark}

%x%\tim{reworded}
All the conditions in (ii) and (iii) which involve $\kappa$ or $\overline W{}^{\lambda}$ can always be satisfied by an appropriate choice of coordinates.
Equivalently, those logarithmic terms which appear if these conditions are not satisfied are pure gauge artifacts.

Recall that to solve the equation for $\xi_A$ both $\kappa$ and $\varphi$ need to be known. This requires a choice of the $\kappa$-gauge. Since the choice of $\overline W{}^0$ does not affect the $\xi_A$-equation, there is no gauge-freedom left in that equation  and if the no-logs-condition \eq{no-log-conditions} does not hold there is no possibility to get rid of the log terms that arise in this equation. (In Section~\ref{sec_no_log} we will return to the question, whether \eq{no-log-conditions} can be satisfied by a choice of $\kappa$.)
Similarly there is no gauge-freedom left when the equation for $\zeta$ is integrated but, due to the special structure of the asymptotic expansion of its source term, no new log terms arise in the expansion of $\zeta$.

The no-logs-condition involves two  functions, $\varphi_{-1}$ and $\varphi_0$, which are globally determined by the gauge function $\kappa$ and the initial data $\gamma$, cf.\  \eq{constraint_phi}.
The dependence of these functions on the gauge and on the initial data is rather intricate.
Thus the question arises for which class of initial data one can find a  function $\kappa=\mathcal{O}(r^{-3})$, such that the no-logs-condition
holds, and accordingly what the geometric restrictions are for this to be possible.
This issue will be analysed in   part II of this work,
using  a gauge scheme  adjusted to the initial data so that all relevant globally defined integration functions can be computed explicitly.

%\ptc{The original NoGoSection is in a file NoGoSectionOrig.tex, input command commented out}
%\input{NoGoSectionOrig}

\section{The no-logs-condition}
\label{sec_no_log}

\subsection{Gauge-independence}
% \ptc{the original no-log-condition section went to AlmostOriginalNoLogsSection.tex}

%\tim{nec-and-suff, cf below}

In this section we show gauge-independence of \eq{no-log-conditions}.
%x%\tim{reworded}
It is shown in paper II \cite{TimAsymptotics} that \eq{no-log-conditions} arises from
integration of the equation for $\xi_A$, which is independent of the gauge functions $W^{\mu}$. Equation \eq{no-log-conditions} is therefore independent of those functions, as well. So the only relevant freedom is that of rescaling the $r$-coordinate parameterizing the null rays.
We therefore need to compute how (\ref{no-log-conditions}) transforms under rescalings of $r$.
For this we consider a  smooth
%\tim{$C^1$? \\ -- \\ ptc: smooth is fine with me}
 coordinate transformation
\begin{equation}
r\mapsto \tilde r=\tilde r(r,x^A)
\;.
\label{rescaling_r}
\end{equation}
Under \eq{rescaling_r} the function $\varphi$ transforms as a scalar.
We have seen above that a necessary condition for the metric to be smoothly extendable across $\scri^+$ is that $\varphi$ has the asymptotic
behaviour
\begin{equation}
\varphi(r,x^A) \,=\, \varphi_{-1}(x^A) r +\varphi_0 +  \mathcal{O}(r^{-1})\;, \quad \text{with} \quad \varphi_{-1}>0
\;.
\label{asympt_phi}
\end{equation}
The transformed $\varphi$ thus takes the form
\begin{eqnarray*}
\tilde\varphi(\tilde r,x^A) &:=&\varphi(r(\tilde r),x^A)\,=\, \varphi_{-1}(x^A) r(\tilde r) +\varphi_0 +  O(r(\tilde r)^{-1})
\;,
\\
\partial_{\tilde r}\tilde\varphi(\tilde r,x^A) &=& \frac{\partial r}{\partial \tilde r}\partial_r\varphi(r(\tilde r),x^A)\,=\,
\frac{\partial r}{\partial \tilde r}\varphi_{-1}(x^A) r(\tilde r) +  \frac{\partial r}{\partial \tilde r}O(r(\tilde r)^{-2})
\;.
\end{eqnarray*}
If we require $\tilde \varphi$ to be of the form \eq{asympt_phi} as well, it is easy to check that we must have
\begin{eqnarray}
 r(\tilde r, x^A) &=& r_{-1}(x^A)\tilde r + r_0 + O(\tilde r^{-1}) \quad \text{and}
\label{asympt_coord_trafo}
\\
 \partial_{\tilde r}r(\tilde r, x^A) &=& r_{-1}(x^A)  + O(\tilde r^{-2})\;, \quad \text{with} \quad r_{-1}>0
\;.
\label{asympt_coord_trafo2}
\end{eqnarray}

We have:

\begin{Proposition}
  \label{P11XII13.1}
 The no-logs-condition (\ref{no-log-conditions}) is invariant under the coordinate transformations \eq{asympt_coord_trafo}-\eq{asympt_coord_trafo2}.
\end{Proposition}

\proof
For the transformation behavior of the expansion coefficients we obtain
\begin{eqnarray*}
 &\varphi_{-1} \,=\, (r_{-1})^{-1}\tilde\varphi_{-1}\;, \quad \varphi_0 \,=\,\tilde  \varphi_0 - r_0(r_{-1})^{-1}\tilde\varphi_{-1}&
\\
&\Longrightarrow \quad \tau_2 \,=\, -2(\varphi_{-1})^{-1}\varphi_0 \,=\,  r_{-1}\tilde\tau_2
+2 r_0
 \;.&
\end{eqnarray*}
Moreover, with \eq{asympt_coord_trafo}-\eq{asympt_coord_trafo2}
%%
%\begin{equation*}
% r(\tilde r)^{-2} = (r_{-1})^{-2} \tilde r^{-2} - 2r_0(r_{-1})^{-3}\tilde r^{-3} + O_1(\tilde r^{-4})
% \;,
%\end{equation*}
%%
we find
\begin{eqnarray*}
 \tilde\sigma_A{}^B &=& \frac{\partial r}{\partial \tilde r} \sigma_A{}^B =\left[ r_{-1} + O(\tilde r^{-2})  \right]\left[ (\sigma_A{}^B)_2r(\tilde r)^{-2} + (\sigma_A{}^B)_3 r(\tilde r)^{-3} + \mathcal{O}(r(\tilde r)^{-4}) \right]
\\
 &=& (r_{-1})^{-1} (\sigma_A{}^B)_2\tilde r^{-2} + \left[ (r_{-1})^{-2}(\sigma_A{}^B)_3 - 2r_0(r_{-1})^{-2}(\sigma_A{}^B)_2 \right]\tilde r^{-3}
+ O(\tilde r^{-4})
\\
 \Longrightarrow && (\sigma_A{}^B)_2 = r_{-1} (\tilde \sigma_A{}^B)_2\;,
\\
&&  (\sigma_A{}^B)_3 = (r_{-1})^2(\tilde\sigma_A{}^B)_3 + 2r_0r_{-1}(\tilde \sigma_A{}^B)_2
 \;.
\end{eqnarray*}
Hence
\begin{equation*}
 (\sigma_A{}^B)_3- \tau_2 (\sigma_A{}^B)_2
=
 (r_{-1})^2[ (\tilde\sigma_A{}^B)_3 - \tilde\tau_2  (\tilde\sigma_A{}^B)_2 ]
 \;.
\end{equation*}
\qed

Although the No-Go Theorem~\ref{T9XII13.11} shows that the ($\kappa=0$, $\ol W^\lambda =0$)-wave-map gauge invariably produces logarithmic terms except in the flat case,
one can decide whether the logarithmic terms can be transformed away by checking  \eq{no-log-conditions}  using this gauge, or in fact any other.
In the $([\gamma],\kappa)$ scheme this requires to determine $\tau_2$ by solving the  Raychaudhuri equation, which makes this scheme not practical for the purpose. In particular, it is not a priori clear within this scheme whether \emph{any} initial data satisfying this condition exist unless both $(\sigma^A{}_B)_2$ and $(\sigma^A{}_B)_3$ vanish.
On the other hand,
in any  gauge scheme where  $\check g$
is prescribed on the cone, the no-logs-condition \eq{no-log-conditions} is a straightforward condition on the asymptotic
behaviour of the metric.

Let us  assume that \eq{no-log-conditions} is violated for say $\kappa=0$. We know that the metric cannot have a smooth conformal completion at infinity in an adapted null coordinate system arising from the $  \kappa=0$-gauge via a transformation which \textit{is not} of the asymptotic form  \eq{asympt_coord_trafo}-\eq{asympt_coord_trafo2}.
On the other hand if the transformation \textit{is} of the form  \eq{asympt_coord_trafo}, then the no-logs-condition will also be violated in the new coordinates. We conclude that we cannot have a smooth conformal completion in \emph{any} adapted null coordinate system. That yields

\begin{theorem}
 \label{T3III14.1}
Consider initial data $\gamma$ on a light-cone $C_O$ in a $ \kappa=0$-gauge
with asymptotic behaviour  $\gamma_{AB}\sim  r^2 ( s_{AB}+ \sum_{n=1}^{\infty} h^{(n)}_{AB} r^{-n})$.
Assume that $\varphi$, $\nu^0$ and $\varphi_{-1}$ are strictly positive on $C_O\setminus\{O\}$ and $S^2$, respectively.
Then there exist %adapted null coordinates $(u,r,x^A)$
a gauge w.r.t.\ which the trace $\ol g$ of the metric on the cone admits a smooth conformal completion at infinity
 and where the connection coefficients $\ol \Gamma^r_{rA}$ are smooth at $\scrip$
 (in the sense of  Definition~\ref{definition_smooth})
if and only if  the no-logs-condition \eq{no-log-conditions} holds
in one (and then any)  coordinate system related to the original one by a coordinate transformation of the form \eq{asympt_coord_trafo}-\eq{asympt_coord_trafo2}.
\end{theorem}

\subsection{Geometric interpretation}
\label{no-logs_geom}

Here we provide a geometric interpretation
of the no-logs-condition  \eq{no-log-conditions} in terms of the conformal Weyl tensor. This ties our results with the analysis in~\cite{andersson:chrusciel:PRL} (compare also Section~\ref{ss16XII13.5}).

For this purpose let us consider the components of the conformal Weyl tensor,  $ C_{rAr}{}^B$, on the cone.
To end up with smooth initial data for the conformal fields equations
we need to require its rescaled counterpart $\overline {\tilde d}_{rAr}{}^B = \ol \Theta^{-1} \ol {\tilde C}_{rAr}{}^B = \ol \Theta^{-1} \ol C_{rAr}{}^B$ to be smooth at $\scri^+$, which is equivalent to
\begin{equation}
 \overline C_{rAr}{}^B = \mathcal{O}(r^{-5})
 \;.
\label{asympt_beh_Weyl}
\end{equation}
In particular the $\ol C_{rAr}{}^B $-components of the Weyl tensor need to vanish one order faster than naively expected from the asymptotic behavior of the metric.
In adapted null coordinates and in vacuum we have, using the formulae of \cite[Appendix~A]{CCM2},
\begin{eqnarray*}
 \ol C_{rAr}{}^B  &=&   \ol R_{rAr}{}^B \,=\,  -\partial_r\ol \Gamma^B_{rA} + \ol\Gamma^B_{rA}\ol \Gamma^r_{rr}
- \ol \Gamma^B_{rC}\ol\Gamma^C_{rA}
\\
&=&   -(\partial_r-\kappa)\chi_A{}^B
 - \chi_A{}^C\chi_C{}^B
\\
&=&   -\frac{1}{2}(\partial_r\tau-\kappa\tau  + \frac{1}{2}\tau^2 ) \delta_A{}^B
 -(\partial_r + \tau -\kappa)\sigma_A{}^B
 - \sigma_A{}^C\sigma_C{}^B
\\
&=& \frac{1}{2}|\sigma|^2\delta_A{}^B  -(\partial_r+\tau -\kappa)\sigma_A{}^B
 - \sigma_A{}^C\sigma_C{}^B
\;.
\end{eqnarray*}
Assuming, for definiteness, that $\kappa=\mathcal{O}(r^{-3})$ and $\ol g_{AB} = \mathcal{O}(r^2)$ with $(\det \ol g_{AB})_{-4}>0$
we  have
\begin{eqnarray*}
 \ol C_{rAr}{}^B  &=&    \Big( (\sigma_A{}^B)_3  - \tau_2 (\sigma_A{}^B)_2 + \frac{1}{2} (\sigma_C{}^D)_2(\sigma_D{}^C)_2\delta_A{}^B
- (\sigma_A{}^C)_2(\sigma_C{}^B)_2\Big) r^{-4}
\\
&& + \mathcal{O}(r^{-5})
\;.
\end{eqnarray*}
 As an $s$-symmetric, trace-free tensor $(\sigma_A{}^{C})_2$ has the property
\begin{equation*}
     (\sigma_A{}^{C})_2(\sigma_C{}^{B})_2 =   \frac{1}{2}(\sigma_D{}^{C})_2(\sigma_C{}^{D})_2\delta_A{}^{B}
\;,
\end{equation*}
i.e.
\begin{eqnarray*}
 \ol C_{rAr}{}^B  &=&    \big[ (\sigma_A{}^B)_3  - \tau_2 (\sigma_A{}^B)_2 \big] r^{-4} + \mathcal{O}(r^{-5})
\;,
\end{eqnarray*}
and \eq{asympt_beh_Weyl}  holds if and only if the no-logs-condition is satisfied.

\section{Other settings}
 \label{s16XII13.1}

We pass now to the discussion, how   to modify the above when other data sets are given, or Cauchy problems other than a light-cone are considered.

\subsection{Prescribed $(\check g_{AB},\kappa)$}
 \label{ss16XII13.1}

In this setting the initial data are a symmetric degenerate twice-covariant tensor field $\check g$, and a connection $\kappa$ on the family of bundles tangent to the integral curves of the kernel of $\check g$, satisfying the Raychaudhuri constraint \eq{constraint_tau}.

Recall that so far we have mainly been considering a characteristic Cauchy problem where $([\gamma], \kappa)$ are given.  There  \eq{constraint_phi} was used to solve for the conformal factor relating $\check g$ and $\gamma$:
\bel{16XII13.1}
 \check g \equiv \overline g_{AB} \mathrm{d}x^A\mathrm{d}x^B = \varphi^2 \big( \frac{\det s_{CD}}{\det \gamma_{EF}}\big)^{\frac{1}{n-1}}
 \gamma_{AB} \mathrm{d}x^A \mathrm{d}x^B
 \;.
\ee
But then a pair $(\check g,\kappa)$ satisfying  \eq{constraint_tau} is
obtained.
So, in fact, prescribing the pair $(\check g,\kappa)$ satisfying  \eq{constraint_tau} can be viewed as a special case of the $([\gamma], \kappa)$-prescription, where one sets $ \gamma:=\check g$.
Indeed, when $\check g$ and $\kappa$ are suitably regular at the vertex, uniqueness of solutions of \eq{constraint_phi} with the boundary conditions $\varphi(0)=0$ and $\partial_r\varphi(0)=1$ shows that %
\bel{18XII13.1}
 \varphi=  \big( \frac{\det \ol g_{EF}}{\det s_{CD}}\big)^{\frac{1}{2(n-1)}}
 \quad
 \Longleftrightarrow
 \quad
 \ol g_{AB}\equiv \gamma_{AB}
 \quad
 \Longleftrightarrow
 \quad
 \check g \equiv \gamma
 \;.
\ee
In particular all the results so far apply to this case.
%\tim{I'm confused here, the boundary conditions used above, $\varphi(0)=0$ and $\partial_r\varphi(0)=1$ are not compatible with this value for $\varphi$ \\ -- \\ ptc: changed }

If $\tau$ is nowhere vanishing, as necessary for a smooth null-geodesically complete light-cone extending to null infinity, then \eq{constraint_tau} can be algebraically solved for $\kappa$, so that the constraint becomes trivial.

\subsection{Prescribed $(\overline g_{\mu\nu},\kappa)$}
 \label{ss16XII13.2}

In this approach one prescribes all   metric functions  $\overline g_{\mu\nu}$ on the initial characteristic hypersurface, together with the connection
coefficient $\kappa$, subject to the Raychaudhuri equation \eq{constraint_tau}.
 \Eq{R11_constraint} relating $\kappa$ and $\nu_0$ becomes an algebraic equation for the gauge-source function $\ol W{}^0$, while the equations $\ol R_{r A}=0=\ol g^{AB}\ol R_{AB}$ become algebraic equations for  $\ol W^A$ and $\ol W{}^r$.

In four space-time dimensions, a smooth conformal completion at null infinity will exist
if and only if $r^{-2}\ol g_{\mu\nu}$  can be smoothly extended as a Lorentzian metric across $\scri^+$ and no logarithmic terms appear in the asymptotic expansion of
 $\ol \Gamma^r_{rA}$;
this last fact is equivalent to   \eq{no-log-conditions}. To see this,
note that since the equations for $\ol W^{\mu}$ are algebraic,
no log terms arise in these fields as long as no log terms appear in the remaining fields appearing in the constraint equation.
Similarly no log terms arise in the $\zeta$-equation.
The only possible source of log terms is thus the $\xi_A$-equation, and the appearance of log terms there is excluded precisely by the no-logs-condition. The existence of an associated space-time with a ``piece of smooth $\scrip$'' follows then from the analysis of the initial data for Friedrich's conformal equations in part II of this work, together with the analysis in \cite{CPW}.

We conclude that \eq{no-log-conditions} is again a necessary-and-sufficient condition for existence of a smooth $\scrip$ for the current scheme in space-time dimension four.

\subsection{Frame components of $\sigma$ as free data}
 \label{ss16XII13.4}

In this section we consider as free data the components $\chi_{ab}$ in an adapted parallel-propagated frame as in~\cite[Section~5.6]{ChPaetz}.
We will assume that
\bel{22XII13.1}
 \chi^a{}_b = \frac 1 r \delta^ a{}_b + \mathcal{O}(r^{-2})
  \;,
  \quad
  a,b \in\{2,3\}
 \;.
\ee
There are actually at least two schemes which would lead to this form of $ \chi^a{}_b $:
One can e.g.\ prescribe any $ \chi^a{}_b $ satisfying
\eq{22XII13.1} such that $\chi^2{}_2+\chi^3{}_3= \chi_{22}+\chi_{33}$ has no zeros, define $\sigma_{ab}=\chi_{ab}- \frac 12 (\chi^2{}_2+\chi^3{}_3) \delta_{ab}$,
and solve algebraically the
Raychaudhuri equation for $\kappa$.
 Another possibility is to prescribe directly a symmetric trace-free tensor $\sigma_{ab}$
in the   $\kappa=0$ gauge, use the Raychaudhuri equation to determine $\tau$, and construct $\chi_{ab}$   using
\bel{22XII13.2}
 \chi^{a}{}_{b}= \frac \tau {2} \delta^ a{}_b + \sigma^ a{}_b
  \;,
  \quad
  a,b \in\{2,3\}
 \;.
\ee
The asymptotics \eq{22XII13.1} will then hold if $\sigma^ a{}_b$ is taken to be $ \mathcal{O}(r^{-2})$.

Given $\chi_{ab}$, the tensor field $\check g$ is obtained by setting
\bel{22XII13.3}
 \check g = \big
  (\theta^2{}_A  \theta^2{}_B
 +  \theta^3{}_A  \theta^3{}_B
  \big)
    \mathrm{d}x^ A \mathrm{d} x^B
  \;,
\ee
where the co-frame coefficients $
  \theta^a{}_A$ are solutions of the equation~\cite{ChPaetz}
\bel{22XII13.4-}
   \partial_r   \theta^a{}_A  = \chi^a{}_b  \theta^b{}_A
  \;,
  \quad
  a,b \in \{2,3\}
  \;.
\ee

Assuming \eq{22XII13.1}, one finds that solutions of \eq{22XII13.4-} have an asymptotic expansion for $\theta^a{}_A$ without log terms:
\bel{22XII13.4+}
 \theta^a{}_A = r \varphi^a{}_A + \mathcal{O}(1)
  \;,
  \quad
  a,b \in \{2,3\}
\ee
for some globally determined functions $\varphi^a{}_A$. If the determinant of the two-by-two matrix $(  \varphi^a{}_A )$ does not vanish, one obtains a tensor field $\check g$ to which our previous considerations apply. This leads again to the no-logs-condition (\ref{no-log-conditions}).

Writing, as usual,
\bel{22XII13.5}
 \sigma_{ab} = (\sigma_{ab})_{2} r^{-2} +
 (\sigma_{ab})_{3} r^{-3} + \mathcal{O}(r^{-4})
  \;,
  \quad
  a,b \in \{2,3\}
  \;,
\ee
the no-logs-condition rewritten in terms of $\sigma_{ab}$ reads
\bel{22XII13.6}
   (\sigma_{ab})_3  =  \tau_2(\sigma_{ab})_2
    \;,
  \quad
  a,b \in \{2,3\}
  \;.
\ee

\subsection{Frame components of the Weyl tensor  as free data}
 \label{ss16XII13.5}

Let $C_{\alpha \beta \gamma \delta}$ denote the space-time Weyl tensor.
For $a,b\ge 2$ let
\begin{equation*}
 \psi_{ab} := e_a{}^A e_b {}^B \overline C_{A r B r}
\end{equation*}
represent the components of $ \overline C_{A r B r }$ in a parallelly-transported adapted frame, as in Section~\ref{ss16XII13.4}. The tensor field
$\psi_{ab}$ is symmetric, with vanishing $\eta$-trace, and we have in space-time dimension four (cf., e.g., \cite[Section~5.7]{ChPaetz})
\begin{eqnarray}
 (\partial_r -\kappa)\chi_{ab} & = &
  - \sum_{c=2}^3
    \chi_{ac}\chi_{cb}
        -    \psi_{ab}  -\frac 1{2}\eta_{ab}\overline  T_{rr}
        \;.
\label{22I12.1x2}
\end{eqnarray}
Given $(\kappa, \psi_{ab})$, we can integrate this equation in vacuum to obtain the tensor field $\chi_{ab}$ needed in Section~\ref{ss16XII13.4}.
However, this approach leads to at least two difficulties: First, it is not clear under which conditions on $\psi_{ab}$ the solutions will exist for all values of $r$. Next, it is not clear that the global solutions will have the desired asymptotics. We will not address these questions but, taking into account the behaviour of the Weyl tensor under conformal transformations, we will assume that
\bel{22XII13.11}
\kappa =
 \mathcal{O}(r^{-3})
 \;,
 \quad
 \psi_{ab}=
 \mathcal{O}(r^{-4})
 \;,
\ee
and that the associated tensor field $\chi_{ab}$  exists globally and satisfies \eq{22XII13.1}. The no-logs-condition will then hold if and only if
\bel{22XII13.12}
 \psi_{ab}=
 \mathcal{O}(r^{-5})
  \qquad
  \Longleftrightarrow
  \qquad
  (\psi_{ab})_4= 0
 \;.
\ee

Note that one can reverse the procedure just described: given $\chi_{ab}$ we can use \eq{22I12.1x2} to determine $\psi_{ab}$.
Assuming \eq{22XII13.1}, the no-logs-condition will  hold if and only if the $\psi_{ab}$-components of the Weyl tensor vanish one order faster than naively expected from the asymptotic behaviour of the metric (cf.\ Section~\ref{no-logs_geom}).

 %x%\ptcr{remark added, as hinted at by a referee; his suggestion to study the issue further is noted; we have nothing new to report on this so far, but we will keep trying}
\Eq{22XII13.12} is the well-known starting point of the analysis in \cite{Newman:Penrose}, and has also been obtained previously as a necessary condition for existence of a smooth $\scri$ in the analysis of the hyperboloidal Cauchy problem~\cite{andersson:chrusciel:PRL}. It is therefore not surprising that it reappears in the analysis of the characteristic Cauchy problem. However, as pointed out above, a satisfactory treatment of the problem using $\psi_{ab}$ as initial data requires further work.

\subsection{Characteristic surfaces intersecting transversally}
 \label{s10XII13.1}

Consider two characteristic surfaces, say $\mcN_1$ and $\mcN_2$, intersecting transversally along a smooth submanifold $S$ diffeomorphic
to $S^2$.
Assume moreover that the initial data on $\mcN_1$ (in any of the versions just discussed) are such that the metric $\ol g_{\mu\nu}$ admits a smooth conformal completion across the sphere $\{x=0\}$, as in Definition~\ref{definition_smooth}. The no-logs-condition  \eq{no-log-conditions} remains unchanged. Indeed, the only difference is the integration procedure for the constraint equations:
while on the light-cone we have been integrating from the tip of the light-cone, on $\mcN_1$ we   integrate  from the intersection surface~$S$. This leads to the need to provide supplementary data  at $S$  which render the solutions unique. Hence the asymptotic values of the fields, which arise from~the integration of the constraints, will also depend on the supplementary data at $S$.

\subsection{Mixed spacelike-characteristic initial value problem}
 \label{s10XII13.2}

Consider a mixed initial value problem, where the initial data set consists of:

\begin{enumerate}
\item
    A spacelike initial data set $(\hyp,{}^3g,K)$, where ${}^3g$
is a Riemannian metric on $\hyp$ and $K$ is a symmetric two-covariant tensor field on $\hyp$. The three-dimensional manifold $\hyp$ is supposed to have a compact smooth boundary $S$ diffeomorphic to $S^2$, and the fields $( {}^3g,K)$ are assumed to satisfy the usual vacuum Einstein constraint equations.

\item A hypersurface $\mcN_1$ with boundary $S$ equipped with a  characteristic initial data set, in any of the configurations discussed so far. Here $\mcN_1$ should be thought of as a characteristic initial data surface emanating from $S$ in the outgoing direction.
    \item The data on $\hyp$ and $\mcN_1$ satisfy a set of ``corner conditions'' at $S$, to be defined shortly.
\end{enumerate}

The usual evolution theorems for the spacelike general relativistic initial value problem provide a unique future maximal globally hyperbolic vacuum development $\mcD^+$ of  $(\hyp,{}^3g,K)$. Since $\hyp$ has a boundary, $\mcD^+$ will also have a boundary. Near $S$, the null part of the boundary of $\partial \mcD^+$ will be a smooth null hypersurface emanating from $S$, say $\mcN_2$, generated by null geodesics normal to $S$ and ``pointing towards $\hyp$'' at $S$.
In particular the space-time metric on $\mcD^+$  induces characteristic initial data on $\mcN_2$. In fact, all  derivatives of the metric, both in directions tangent and transverse to $\mcN_2$, will be determined on $\mcN_2$ by the  initial data set $(\hyp,{}^3g,K)$. This implies that the characteristic initial data needed on $\mcN_1$, as well as their derivatives in directions tangent to $\mcN_1$, are determined on $S$ by $(\hyp,{}^3g,K)$. These are the ``corner conditions'' which have to be satisfied by the data   on $\mcN_1$ at $S$, with these data being arbitrary otherwise. The corner conditions can be calculated algebraically in terms of the fields $( {}^3g,K)$,   the gauge-source functions $W^\mu$, and the  derivatives of those fields, at $S$, using the vacuum Einstein equations.

One can use now the  Cauchy problem  discussed in Section~\ref{s10XII13.1} to obtain the metric to the future of $\mcN_1\cup \mcN_2$, and the discussion of the no-logs-condition given in Section~\ref{s10XII13.1} applies.